\documentclass[11pt,a4paper]{article}

\usepackage{mathtools}
\usepackage{authblk} 
\usepackage{algpseudocode,algorithmicx,algorithm}
\usepackage{subfigure}
\usepackage{array}
\usepackage{booktabs}

\usepackage{mathrsfs}
\usepackage{latexsym,bm}
\usepackage{amsmath,amsfonts,amsmath,amssymb,amsthm,bbm}
\usepackage{extarrows}

\usepackage{graphicx,subfigure,epstopdf,float}
\usepackage{enumerate,cases,multirow}
\usepackage{caption}

\usepackage{longtable,colortbl,arydshln,threeparttable}
\definecolor{mygray}{gray}{.9}

\usepackage{indentfirst}
\usepackage[top=25mm,bottom=25mm,left=20mm,right=20mm]{geometry}
\usepackage{cite}

\usepackage{listings}

\usepackage{makeidx}        
\usepackage{booktabs}
\usepackage[unicode=true,bookmarksnumbered,colorlinks,citecolor=red,linkcolor=red,hyperindex,linktocpage=true]{hyperref}
\usepackage{pifont}

\newcommand{\ket}[1]{| #1 \rangle} 

\newcommand{\bb}{\boldsymbol}

\def \d {\mathrm{d}}

\def \i {\mathrm{i}}

\newcounter{parentalgorithm}

\makeatother

\newtheorem{theorem}{Theorem}[section]
\newtheorem{lemma}{Lemma}[section]

\theoremstyle{remark}
\newtheorem{remark}{\bf Remark}[section]

\numberwithin{equation}{section}
\usepackage{pifont}

\begin{document}

\title{Quantum simulation of multiscale linear transport equations via Schr\"odingerization and exponential integrators}
\author[1,2,3]{Xiaoyang He\thanks{hexiaoyang@sjtu.edu.cn}}
\author[1,2,3]{Shi Jin\thanks{shijin-m@sjtu.edu.cn}\thanks{Corresponding Author.}}
\affil[1]{School of Mathematical Sciences, Shanghai Jiao Tong University, Shanghai, 200240, China}
\affil[2]{Shanghai Center for Applied Mathematics (SJTU Center), Shanghai Jiao Tong University, Shanghai, 200240, China}
\affil[3]{Institute of Natural Sciences, Shanghai Jiao Tong University, Shanghai, 200240, China}
\date{}
\maketitle

\begin{abstract}
\par In this paper, we present two Hamiltonian simulation algorithms for multiscale linear transport equations, combining the Schr\"odingerization method [S. Jin, N. Liu and Y. Yu, Phys. Rev. Lett, 133 (2024), 230602][S. Jin, N. Liu and Y. Yu, Phys. Rev. A, 108 (2023), 032603] and exponential integrator while incorporating incoming boundary conditions. These two algorithms each have advantages in terms of design easiness  and scalability, and the query complexity of both algorithms, $\mathcal{O}(N_vN_x^2\log N_x)$, outperforms existing quantum and classical algorithms for solving this equation. In terms of the theoretical framework, these are the first quantum Hamiltonian simulation algorithms for multiscale linear transport equation to combine the Schr\"odingerization method with an effective asymptotic-preserving schemes, which are efficient for handling multiscale problems with stiff terms. 
\end{abstract}


\tableofcontents

\maketitle
\section{Introduction}
\par The linear transport equation is a type of integro-partial differential equation defined in phase space, used to simulate the probability distribution of particles in a background medium. It has extensive applications in important fields such as astrophysics, nuclear reactions, and medical imaging\cite{CaseZweifel,Chander,MRS,Ryzhik}. However, the classical computation of transport equation is challenging due to the curse of dimensionality and possible multiple time and spatial scales which needs to be resolved numerically, leading to significant computational costs. 
\par The advantages of quantum computing in scientific computation are increasingly being recognized by scientists. Over the years, this field has continuously evolved, yielding a series of outstanding achievements. Harrow, Hassidim, and  Lloyd\cite{Harrow2009QuantumAF} proposed the HHL algorithm, which provides a fundamental method for solving linear equations using quantum computing. The quantum linear algebra algorithms can be extended to solve linear ordinary and partial differential equations after their spatial and/or temporal discretizations
\cite{JLY2022multiscale,BerryChilds2017ODE,Childs2017QLSA,Costa2021QLSA,Subasi2019AQC,Harrow2009QuantumAF}, including linear partial differential equations (PDEs) with multiple scales \cite{Jin2022TimeCA}.\cite{Berry2015Hamiltonian} introduced a method for Hamiltonian simulation using quantum walk algorithms. Hamiltonian simulation is one of the most promising directions for demonstrating quantum advantage early on. When performing simulations on a quantum computer, the Hamiltonian $H$ can be decomposed into a sum of tensor products of Pauli matrices, and then the product formula can be used to approximate the Hamiltonian operator.
\par One difficulty in solving linear PDEs using Hamiltonian simulation is how to transform a non-Hermitian system into a Hermitian Hamiltonian. Recently, two approaches  have been proposed to address this issue. An et al.\cite{An2023QuantumAF} converted non-unitary dynamics into a linear combination of Hamiltonian simulation problems (LCHS) and demonstrated the advantages of this method in terms of state preparation costs.  Jin et al.\cite{Jin2024Schrodingerization,Jin2023Detailes} proposed an algorithm that maps linear differential equations to a higher-dimensional unitary dynamics, known as the Schr\"dingerization method. Recently, the Schr\"odingerization method has demonstrated its powerful universality in its  field of computational science, for both qubit and continuous platforms. It applications include discrete iterative schemes in linear algebra or dynamical systems \cite{Jin2024QuantumSO}, time-dependent PDEs\cite{Cao2023QuantumSF}, physical boundary problems\cite{Jin2024PhysicalBoundary,Jin2024QuantumSF}, and multiscale equations\cite{Hu2024QuantumMultiscale}, but also in its capability to solve specific differential equation problems, such as the Maxwell equations\cite{Jin2023QuantumSO}, the Fokker-Planck equations \cite{Jin2024FokkerPlanck}, the stochastic differential equations \cite{Jinweiwei2025}, the highly oscillatory transport equations \cite{Gu2025HighlyOscillatory}, the Black-Scholes equations\cite{Jin2025BlackScholes}, and the elastic wave equations\cite{Jin2025ElasticWave}, etc. Due to the aforementioned computational challenges,  exploring how to use the Schr\"odingerization method and Hamiltonian simulation to handle linear transport equations is an important yet feasible problem.

\par In this paper, we consider the following  linear multiscale transport equation under diffusive scaling\cite{Jin2000DiffusiveRelaxation}:
\begin{equation}
    \label{equ:transport}
    \varepsilon\partial_t f(\bb{v}) + \bb{v} \cdot \nabla_x f(\bb{v}) = \frac{1}{\varepsilon}\left(\frac{\sigma_S}{S} \int_{\Omega} f(\bb{v}') d \bb{v}' - \sigma f(\bb{v})\right) + \varepsilon Q,
\end{equation}
where $f(t,\bb{x},\bb{v})$ represents the particle probability density distribution. Here, $\bb{x}\in\mathbb{R}^d$ denotes position, $t$ is time, and $\bb{v} \in \Omega\subset\mathbb{R}^d$ corresponds to the cosine of the angle between the velocity and position vectors, with $\int_{\Omega} d\bb{v}=S$ and $\Omega$ being symmetric in $\bb{v}$. The coefficients $\sigma=\sigma(\bb{x})$ and $\sigma_S=\sigma_S(\bb{x})$ are the total transport and scattering coefficients, respectively, while $Q=Q(\bb{x})$ is the source term, and $\varepsilon$ represents the mean free path. Typically, $\sigma_S=\sigma-\varepsilon^2\sigma_A$, where $\sigma_A=\sigma_A(\bb{x})$ is the absorption coefficient. Note that, the scaling in \eqref{equ:transport} is standard and yields a diffusion equation in the limit $\varepsilon \to 0$.
\par Earlier, quantum algorithms for multiscale linear transport equations have already been proposed \cite{Xiaoyang2023TimeCA}. He et al. transformed the one-dimensional linear transport equation into a linear system using a finite difference scheme and solved it using the HHL algorithm. It achieved a query complexity of $\mathcal{O}(N_v^2N_x^2\ln N_x)$, where $N_x$ is the number of spatial discretizations. This represents an exponential improvement over the classical algorithm's time complexity of $\mathcal{O}(N_v^2N_x^3)$, highlighting the advantage of quantum algorithms in computing this equation. However, the deep circuits and noise sensitivity of HHL algorithm make it impractical for NISQ devices, the currently typical quantum hardware platforms, while Hamiltonian simulation is more feasible \cite{Preskill2018}. Therefore, designing quantum algorithms for multiscale linear transport equations based on Hamiltonian simulation remains worthwhile to explore. 
\par However,  applying the Schr\"odingerization method straightforwardly  does not lead to efficient solvers for multiscale linear transport equations. This is because traditional Schr\"odingerization algorithms often employ spatial discretization schemes similar to forward schemes, which cannot adequately handle stiff terms and thus fail to solve multiscale problems efficiently with severe time-step restrictions and consequantly,   the query complexity of such algorithms may depend on the small physical scales $\varepsilon$, thus inefficient even for quantum algorithms.  Our goal to addressing the multiscale problems aiming to achieve an computational  efficiency {\it independent of} $\varepsilon$.  This requires eliminating the impact of stiff terms. The AP (Asymptotic Preserving) scheme-a numerical method designed to preserve the asymptotic limit from microscopic to macroscopic models within a discretized framework\cite{Jin2000DiffusiveRelaxation,Hu2017AP}-can effectively resolve this issue, thus will be applied together with Schr\"odingerization. 

We present two new quantum simulation algorithms for solving multiscale transport equations, each offering distinct advantages. Building upon exponential integrators and Schrödingerization techniques \cite{Jin2024QuantumSO}, our first approach (Section \ref{sec:The iterative methods}) provides an efficient iterative solution to the discrete system 
$x^{n+1} = Ax^n+b$. Our second algorithm (Section \ref{sec:Steady-state solution}) reformulates the problem as a linear system $Hu=b$, solved through Hamiltonian simulation of the corresponding dynamical system $du/dt=b-Hu$.In this section, a key technical breakthrough involves developing a Laplace transform-based method for accurate matrix exponential evaluation, which overcomes the limitations of conventional norm estimates and enables tighter error control. Remarkably, both methods achieve an optimized complexity of $\mathcal{O}(N_vN_x^2logN_x)$, outperform any existing classical and quantum algorithms for the same problem.  

The paper unfolds as follows: Section \ref{sec:Exponential integrators} lays the theoretical foundation by developing the exponential integrator framework and presenting the core discretization scheme. Section \ref{sec:The iterative methods} details the first quantum algorithm, including the iterative solver architecture and its complexity analysis. Section \ref{sec:Steady-state solution} derives the enhanced linear system formulation and demonstrate its complexity analysis. Section \ref{sec:Numerical examples} then provides comprehensive numerical validation and Section \ref{sec:Conclusions and Discussions} concludes with broader perspectives.

\section{Exponential integrators for the linear transport equation}
\label{sec:Exponential integrators}

\par In this section, we will provide a detailed explanation of how to use the exponential integrators method to handle the linear transport equations, transforming them into an iterative scheme of the form $x^{n+1}=Ax^n+b$.

\subsection{Diffusive relaxation system for the linear transport equation}

\par We first consider one-dimensional linear transport equations with constant coefficient by choosing $Q=0$, $\sigma_A=0$, and $\sigma=\sigma_S=1$. The equation is given by
\begin{equation}\label{equ:transport1}
    \varepsilon \partial_{t} f(t,x,v)+v\cdot\nabla_x f(t,x,v)=\frac{1}{\varepsilon}\left(\frac{1}{2} \int_{-1}^1f(t,x,v') dv'-f(t,x,v)\right),
\end{equation}
with the following incoming boundary condition
\begin{equation}\label{incoming boundary}
    \begin{aligned}
    &f(t,x_L,v)=F_L(v),\quad v>0,\\
    &f(t,x_R,v)=F_R(v),\quad v<0.
    \end{aligned}
\end{equation}
\par One of the computational challenges  of solving this equation stems from handling the stiff terms in the convection and collision terms. To address this issue, many AP formulation have been proposed for handling time-dependent kinetic and hyperbolic equations, such as the diffusive relaxation scheme \cite{Jin2000DiffusiveRelaxation} and the micro-macro decomposition \cite{Lemou2008APMM}.

In this article, we  use the diffusive relaxation method to handle Eq. (\ref{equ:transport1}). Note that one can equivalently write Eq. (\ref{equ:transport1}) for the cases $v>0$ and $v<0$ respectively:
\begin{equation}\label{equ:2}
\begin{aligned}
\varepsilon \partial_{t} f(t,x,v)+v\partial_x f(t,x,v) &=\frac{1}{\varepsilon}\left(\frac{1}{2}\int_{-1}^1 f(t,x,v)dv-f(t,x,v)\right), \\
\varepsilon \partial_{t} f(t,x-v)-v\partial_x f(t,x,-v) &=\frac{1}{\varepsilon}\left(\frac{1}{2}\int_{-1}^1 f(t,x,v)dv-f(t,x,-v)\right) .
\end{aligned}
\end{equation}
Introducing the even-parity $r(t,x,v)=[f(t,x,v)+f(t,x,-v)]/2$ and odd-parity $j(t,x,v)=[f(t,x,v)-f(t,x,-v)]/(2\varepsilon)$, then one has the following system:
\begin{equation*}
\begin{aligned}\label{equ:5}
&\partial_{t} r(t,x,v)+v \partial_x j(t,x,v) =\frac{1}{\varepsilon^2}(\rho(t,x)-r(t,x,v)), \\
&\partial_{t} j(t,x,v)+\frac{1}{\varepsilon^2} v\cdot\partial_x r(t,x,v) =-\frac{1}{\varepsilon^2} j(t,x,v),
\end{aligned}
\end{equation*}
where $\rho(t, x)$, the mass density function, is defined as $\rho(t, x)=\int_0^1r(t,x,v)dv$. The method proposed by Jin et al.\cite{Jin2000DiffusiveRelaxation} is a system of non-stiff hyperbolic equations with stiff terms on the right-hand side, referred to as the {\it diffusive relaxation system}
\begin{equation}
\begin{aligned}
\label{equ:6}
&\partial_{t} r(t,x,v)+v \partial_x j(t,x,v) =-\frac{1}{\varepsilon^2}(r(t,x,v)-\rho(t,x)) ,\\
&\partial_{t} j(t,x,v)+\phi v \partial_x r(t,x,v) =-\frac{1}{\varepsilon^2}\left(j(t,x,v)+\left(1-\varepsilon^2 \phi\right) v \partial_x r(t,x,v)\right),
\end{aligned}
\end{equation}
where $\phi=\phi(\varepsilon)$ with $0 \le \phi \le 1 / \varepsilon^{2}$. This condition ensures that both $\phi(\varepsilon)$ and $1 - \varepsilon^2 \phi(\varepsilon)$ remain positive, guaranteeing uniform stability of the problem for small $\varepsilon$. Generally, $\phi(\varepsilon)$ is chosen as $\min\{1,1/\varepsilon\}$.  Since our focus is primarily on the regime where $\varepsilon \ll 1$, we choose $\phi=1$. Subsequently, one further decomposes Eq. (\ref{equ:6}) into the following relaxation step and transport step:

\begin{itemize}
    \item \textit{Relaxation step.}
    \begin{equation}\label{equ:the relaxation step}
    \begin{aligned}
    &\partial_{t} r(t,x,v) =-\frac{1}{\varepsilon^2}(r(t,x,v)-\rho(t,x)),
    \quad\quad\quad\quad\,\,(a1)\\
    &\partial_{t} j(t,x,v) =-\frac{1}{\varepsilon^2}\left(j+\left(1-\varepsilon^2\right) v \partial_x r(t,x,v)\right);
    \quad(a2)
    \end{aligned}
    \end{equation}
    \item  \textit{Convection step.}
    \begin{equation}\label{equ:the transport step}
    \begin{aligned}
    &\partial_{t} r(t,x,v)+v \partial_x j(t,x,v)=0,\quad\, (b1)\\
    &\partial_{t} j(t,x,v)+v \partial_x r(t,x,v)=0.\quad\,\, (b2)
    \end{aligned}
    \end{equation}
\end{itemize}

\begin{remark}
Although we do not directly solve the general form mentioned in Eq. (\ref{equ:transport}) here, the methods we use and the algorithms we propose remain applicable to higher-dimensional cases with variable cofficients. For details, please refer to section \ref{section:the general case of transport equations}.
\end{remark}

\subsection{Exponential integrators of the linear transport equation}

\par To solve the diffusion relaxation scheme mentioned in Eqs. (\ref{equ:the relaxation step}) and (\ref{equ:the transport step}) using discrete numerical methods, we need to adopt an AP (Asymptotic Preserving) discretization scheme. This ensures that the numerical scheme can maintain the asymptotic limit from microscopic models to macroscopic models. Currently, many methods have been proposed for multiscale linear transport equations \cite{Jin2000DiffusiveRelaxation,
Hu2017AP}. Here, we use the exponential integration method proposed by \cite{Pareschi_Russel1999}, where the relaxation step is rearranged into an equivalent form as follows
\begin{equation}\label{equ:exp rela step}
\begin{aligned}
&\partial_{t}[(r(t,x,v)-\rho(t,x))e^{\frac{t}{\varepsilon^2}}] =0,\\
&\partial_t[(j(t,x,v)+v\partial_x\rho(t,x))e^{\frac{t}{\varepsilon^2}}] = 
-\frac{1-\varepsilon^2}{\varepsilon^2}e^{\frac{t}{\varepsilon^2}}v\partial_xr(t,x,v)
+\frac{1}{\varepsilon^2}e^{\frac{t}{\varepsilon^2}}v\partial_x\rho(t,x), 
\end{aligned}
\end{equation}
in which the first equation can be directly derived from Eq. (\ref{equ:the relaxation step})(a1), while the second equation utilizes the combination of Eq. (\ref{equ:the relaxation step})(a2), Eq. (\ref{equ:exp rela step})(a1), and $\partial_t\rho(t,x)=0$ ($\rho$ is independent of $t$ during the collision). This formula is obtained by integrating Eq. (\ref{equ:the relaxation step})(a1) over the interval $[0,1]$ with respect to $v$ and using the definition of $\rho(t,x)=\int_0^1r(t,x,v)dv$.
\par Next, we discretize the equivalent form of the relaxation step proposed in Eq. (\ref{equ:exp rela step}) and the convection step in Eq. (\ref{equ:the transport step}). We denote $n$, $m$, and $k$ as the discretization indices of $t$, $x$, and $v$, respectively, and let $r_{k,m}^*$ and $j_{k,m}^*$ be the intermediate results of the relaxation step. For the discretization of velocity, we employ the discrete-ordinate method, selecting the discrete velocity $v_k$ as Legendre-Gauss quadrature points. Consequently, $\rho(t,x)$ can be integrated using Gauss's quadrature rule. Thus, we obtain the exponential integral discretization schemes for the relaxation step and the convection step:
\begin{enumerate}
\item \textit{Relaxation step.}
By employing the central difference approximation for spatial discretization and the explicit Euler scheme for time integration, we obtain
\begin{equation}
\begin{aligned}\label{equ:the discreted relaxation step}
&r_{k,m}^* = 
e^{-\frac{\tau}{\varepsilon^2}}r_{k,m}^n
+(1-e^{-\frac{\tau}{\varepsilon^2}})
\rho_m^n,\\
&j_{k,m}^* = 
e^{-\frac{\tau}{\varepsilon^2}}
j_{k,m}^n
-(1-e^{-\frac{\tau}{\varepsilon^2}}-\frac{\tau}{\varepsilon^2}e^{-\frac{\tau}{\varepsilon^2}})
v_k
\frac{\rho_{m+1}^n-\rho_{m-1}^n}{2h}
-
\frac{\tau}{\varepsilon^2}
e^{-\frac{\tau}{\varepsilon^2}}
(1-\varepsilon^2)
v_k
\frac{r_{m+1}^n-r_{m-1}^n}{2h}.
\end{aligned}
\end{equation}
\item \textit{Convection step.}
By employing the first order upwind method, we obtain
\begin{equation}
\begin{aligned}\label{equ:the discreted convection step}
&r_{k,m}^{n+1} = r_{k,m}^{*}
-\frac{\tau v_k}{2h}
(j_{k,m+1}^{*}-j_{k,m-1}^{*})
+\frac{\tau v_k}{2h}
(r_{k,m+1}^{*}-2r_{k,m}^{*}+r_{k,m-1}^{*}),\\
&j_{k,m}^{n+1} = j_{k,m}^{*}
-\frac{\tau v_k}{2h}
(r_{k,m+1}^{*}-r_{k,m-1}^{*})
+\frac{\tau v_k}{2h}
(j_{k,m+1}^{*}-2j_{k,m}^{*}+j_{k,m-1}^{*}).
\end{aligned}
\end{equation}
The detailed procedure can be found in \cite{Jin2000DiffusiveRelaxation}. Note that the boundary treatment is not considered here and will be covered in the next subsection.
\end{enumerate}

\subsection{Incoming boundary condition}

\par Previous studies on multiscale linear transport equations\cite{Xiaoyang2023TimeCA} have considered zero boundary conditions, but this does not fully represent real-world application cases. In this section, we will discuss how to handle boundary conditions as in \cite{Jin2000DiffusiveRelaxation}. First, we provide the boundary conditions for $r(t,x,v)$ and $j(t,x,v)$ as follows
\begin{equation}\label{boundary:rj}
(r(t,x,v)+\varepsilon j(t,x,v))|_{x= x_L} = F_L(v),\qquad
(r(t,x,v)-\varepsilon j(t,x,v))|_{x= x_R} = F_R(v).
\end{equation}
In particular, when $\varepsilon \ll 1$, we use the approximation $j(t,x,v) \approx -v\partial_x r(t,x,v)$ for steady-state solutions. Substituting this into \eqref{boundary:rj} yields
\begin{equation}\label{boundary:r}
(r(t,x,v)-\varepsilon v\partial_x r(t,x,v))|_{x= x_L} = F_L(v),\qquad
(r(t,x,v)+\varepsilon v\partial_x r(t,x,v))|_{x= x_R} = F_R(v).
\end{equation}
We then discretize Eq. (\ref{boundary:r}) and obtain the boundary conditions for $r(t,x,v)$ as
\begin{align*}
&r_{k,0} = \frac{\varepsilon v_k}{\varepsilon v_k+h}r_{k,1}+\frac{h}{\varepsilon v_k+h}F_L(v_k),\quad r_{k,N_x+1} = \frac{\varepsilon v_k}{\varepsilon v_k+h}r_{k,N_x}+\frac{h}{\varepsilon v_k+h}F_R(v_k),
\end{align*}
while the boundary conditions for $j$ is
\begin{align*}
&j_{k,0}
=
-\frac{ v_k}{\varepsilon v_k+h}r_{k,1}
+\frac{ v_k}{\varepsilon v_k+h}F_L(v_k),\quad j_{k,N_x+1}
=\frac{ v_k}{\varepsilon v_k+h}r_{k,N_x}
-\frac{ v_k}{\varepsilon v_k+h}F_R(v_k).
\end{align*}
Then one can get the boundary conditions for $\rho$.
Next, we incorporate the aforementioned incoming boundary conditions to derive the matrix form of the relaxation step for the exponential integrators in Eq. (\ref{equ:the discreted relaxation step}). Specifically, for a fixed velocity $v_k$, we define the vectors  
$\mathbf{r}_k^n=[r_{k,1}^n,\cdots,r_{k,N_x}^n]$ \text{and} $\mathbf{j}_k^n=[j_{k,1}^n,\cdots, j_{k,N_x}^n]$ (same to $\mathbf{r}_k^*$ and $\mathbf{j}_k^*$). Subsequently, we assemble the global matrices  
$\mathbf{r}^n=[\mathbf{r}_1^n;\cdots;\mathbf{r}_{N_v}^n]$ and $\mathbf{j}^n=[\mathbf{j}_1^n;\cdots;\mathbf{j}_{N_v}^n]$ (same to $\mathbf{r}^*$ and $\mathbf{j}^*$), yielding the following compact matrix representation:
\begin{equation}\label{equ:r*j*}
\begin{aligned}
&\bb{r}^* = \left[\beta_1+(1-\beta_1)(W\otimes I)\right]\bb{r}^n,\\
&\bb{j}^* = \beta_1\bb{j}^n-
\left[\beta_2
(VW\otimes I)
\nabla_{\varepsilon,x}+
\beta_3(V\otimes I)\nabla_{\varepsilon,x}\right]
\bb{r}^n
-
\left[\beta_2
(VW\otimes I)+\beta_3(V\otimes I)\right]B_{v}^n,
\end{aligned}
\end{equation}
where the parameters are defined as
\begin{equation*}
\beta_1 = e^{-\frac{\tau}{\varepsilon^2}},\quad \beta_2 = 1-\left(1+\frac{\tau}{\varepsilon^2}\right)e^{-\frac{\tau}{\varepsilon^2}},\quad
\beta_3 = \frac{\tau}{\varepsilon^2}e^{-\frac{\tau}{\varepsilon^2}}(1-\varepsilon^2),
\end{equation*}
and $V = \text{diag}([v_1,v_2,\cdots,v_{N_v}])$, $W = \bb{1}_{N_v}\otimes [w_1,w_2,\cdots,w_{N_v} ]$, $\nabla_{\varepsilon,x}=\text{diag}[\nabla_{\varepsilon,x}^1,\cdots,\nabla_{\varepsilon,x}^{N_v}]$, $B^{n} = 
[B_1^{n};\cdots;B_{N_v}^{n}]$, with $\nabla_{\varepsilon,x}^k$ and $B_k^{n}$ defined as
\begin{equation*}
\nabla_{\varepsilon,x}^k =
\frac{1}{2h}
\begin{bmatrix}
\frac{-\varepsilon v_k}{\varepsilon v_k+h}&1&& & \\
 -1  &  0 &  1&    & \\
& \ddots&\ddots&\ddots & \\
&& -1& 0 &1 \\
& & & -1 &\frac{\varepsilon v_k}{\varepsilon v_k+h} 
\end{bmatrix},
\quad
B_k^{n} = \frac{1}{2(\varepsilon v_k+h)}
\begin{bmatrix}
-F_L(v_k)\\
0\\
\vdots\\
0\\
F_R(v_k)
\end{bmatrix}.
\end{equation*}
Furthermore, we can simplify the form of $\nabla_{\varepsilon,x}$ as $\nabla_{\varepsilon,x}=\frac{1}{2h}(I\otimes D_h+E_\nabla\otimes I_1)$ by defining it through the following matrix:
\begin{equation*}
 D_h = 
\begin{bmatrix}
0  &  1       &                &    &       \\
-1  &  0        &  1        &    &        \\
& \ddots    &   \ddots       & \ddots  &   \\
&           &    -1      & 0  &    1    \\
&           &                &  -1  &0
\end{bmatrix},
I_1=\begin{bmatrix}
-1 &        &                &    &       \\
  &  0        &         &    &        \\
&     &   \ddots       &  &   \\
&           &         &   0 &       \\
&           &                &  &1
\end{bmatrix},
E_{\nabla} =
\begin{bmatrix}
\frac{\varepsilon v_1}{\varepsilon v_1+h} &       &                &    &       \\
 &  \frac{\varepsilon v_2}{\varepsilon v_2+h}        &         &    &        \\
&     &   \ddots       &   &   \\
&           &                 &\frac{\varepsilon v_{N_v}}{\varepsilon v_{N_v}+h}
\end{bmatrix},
\end{equation*}
where $D_h$ is the central difference matrix. Note that the discretization error of this step in the $x$ direction is $\mathcal{O}(h)$. 

We verify whether the numerical scheme for the relaxation step is AP, that is,  whether the discrete scheme can reproduce the discrete limit of the original problem when $\varepsilon \to 0$. Under this condition, the three parameter set satisfy: $\beta_1 \rightarrow 0$,
$\beta_2 \rightarrow 1$ and
$\beta_3 \rightarrow 0$, yielding
\eqref{equ:r*j*} gives 
\begin{equation} \label{equ:r*j*_limit}
\begin{aligned}
\bb{r}^* &= \bb{\rho}^n + o(\varepsilon^2),\\
\bb{j}^* &= -(V\otimes 
\bb{\mathcal{D}_h})\bb{\rho}^n
+o(\varepsilon^2),
\end{aligned}
\end{equation}
where $\bb{\mathcal{D}_h}\bb{\rho}^n$ denotes the central difference approximation of $\bb{\rho}$. Substituting \eqref{equ:r*j*_limit} into the explicit convection step ensures the correct discrete  Euler limit for $r$ and $j$, demonstrating that \eqref{equ:r*j*} is strongly asymptotically preserving (AP) in time.
\par Similarly, under the condition of the incoming boundary condition, we can obtain the matrix form of the discretization scheme for the convection step presented in Eq. (\ref{equ:the discreted convection step}) as follows
\begin{equation}\label{equ:r_n+1j_n+1}
\begin{aligned}
&\bb{r}^{n+1} =  (I+\frac{\lambda}{2} V\otimes L_{h}^{\mathrm{bdry}})\bb{r}^* - \frac{\lambda}{2}V\otimes D_h\bb{j}^*
+ \frac{\lambda}{2}\bb{f}^*, \\
&\bb{j}^{n+1} = (I+\frac{\lambda}{2}V\otimes L_h) \bb{j}^* - \frac{\lambda}{2} V\otimes D_{h}^{\mathrm{bdry}}\bb{r}^* + \frac{\lambda}{2} \bb{g}^*,
\end{aligned}
\end{equation}
in which $\lambda=\frac{\tau}{h}$, $L_{h}^{\mathrm{bdry}}=\text{diag}[L_{1,h}^{\mathrm{bdry}};\cdots;L_{N_v,h}^{\mathrm{bdry}}]$ and $D_{h}^{\mathrm{bdry}}=[D_{1,h}^{\mathrm{bdry}};\cdots;D_{N_v,h}^{\mathrm{bdry}}]$, which incorporate boundary effects, are defined as
\begin{equation*}
L_{k,h}^{\mathrm{bdry}} = L_h+\text{diag}\left[\frac{\varepsilon v_k-v_k}{\varepsilon v_k+h};0;\cdots;0;\frac{\varepsilon v_k-v_k}{\varepsilon v_k+h}\right],\quad
D_{k,h}^{\mathrm{bdry}} = D_h+\text{diag}\left[\frac{v_k-\varepsilon v_k}{\varepsilon v_k+h} ;0;\cdots;0;-\frac{v_k-\varepsilon v_k}{\varepsilon v_k+h}\right],
\end{equation*}
Similarly, we can simplify the form of $L_{h}^{\mathrm{bdry}}$ as $L_{h}^{\mathrm{bdry}}=I\otimes L_h+\frac{\varepsilon-1}{\varepsilon}E_\nabla\otimes I_2$ and $D_{h}^{\mathrm{bdry}}$ as $D_{h}^{\mathrm{bdry}}=I\otimes D_h+\frac{\varepsilon-1}{\varepsilon}E_\nabla\otimes I_1$, by defining it through the following matrix:
\begin{equation*}
L_h =
\begin{bmatrix}
-2 &  1       &                &    &       \\
1  &  -2        &  1        &    &        \\
& \ddots    &   \ddots       & \ddots  &   \\
&           &    1      & -2  &    1    \\
&           &                &  1 &-2
\end{bmatrix},
\quad
I_2=\begin{bmatrix}
1 &        &                &    &       \\
  &  0        &         &    &        \\
&     &   \ddots       &  &   \\
&           &         &   0 &       \\
&           &                &  &1
\end{bmatrix},
\end{equation*}
where $L_h$ is the second derivative matrix. The boundary conditions are given by $\bb{f}^{*} = [\bb{f}_{1}^{*};\cdots;\bb{f}_{N_v}^{*}]$, $\bb{g}^{*} = [\bb{g}_{1}^{*};\cdots;\bb{g}_{N_v}^{*}]$
with
\begin{equation*}
\bb{f}_{k}^{*} = 
\begin{bmatrix}
\frac{hv_k+v_k^2}{\varepsilon v_k+h}F_L(v_k);
0;\cdots;0;
\frac{hv_k+v_k^2}{\varepsilon v_k+h}F_R(v_k)
\end{bmatrix},\quad
\bb{g}_{k}^{*} = 
\begin{bmatrix}
\frac{hv_k+v_k^2}{\varepsilon v_k+h}F_L(v_k);
0;\cdots;0;
-\frac{hv_k+v_k^2}{\varepsilon v_k+h}F_R(v_k)
\end{bmatrix}.
\end{equation*}
\par Finally, we combine the discretization of the relaxation step in Eq. (\ref{equ:r*j*}) with the discretization of the convection step in Eq. (\ref{equ:r_n+1j_n+1}), eliminating the intermediate variables $r^*$ and $j^*$. Through this, we can write the discrete form of the overall exponential integration method at the discrete time $n$ as follows:
\begin{equation}\label{equ:iter_sys_1}
\begin{aligned}
\bb{r}^{n+1} = 
A_1\bb{r}^n+B_1\bb{j}^n+\bb{b}_r,\\
\bb{j}^{n+1}=
A_2\bb{r}^n+B_2\bb{j}^n+\bb{b}_j,
\end{aligned}
\end{equation}
in which the definitions of matrices $A_1$ and $A_2$ are as follows
\begin{equation*}
\begin{aligned}
A_1 
=&
(I+\frac{\lambda}{2} (V\otimes I)(I \otimes L_h+\frac{\varepsilon-1}{\varepsilon}E_\nabla\otimes I_2))
(\beta_1 I+(1-\beta_1)({W}\otimes I))\\
&\quad+
\frac{\lambda}{4h}
\beta_2
((V^2W\otimes D_h)(I\otimes D_h+E_\nabla\otimes I_1))
+
\frac{\lambda}{4h}\beta_3((V^2\otimes D_h)(I\otimes D_h+E_\nabla\otimes I_1))
,\\
A_2 
=&
-(I+\frac{\lambda}{2} V\otimes L_h)
(\frac{\beta_3}{2h}
(V\otimes I)(I\otimes D_h+E_\nabla\otimes I_1)+
\frac{\beta_2}{2h}(VW\otimes I)(I\otimes D_h+E_\nabla\otimes I_1))\\
&\quad- 
\frac{\lambda}{2}((V\otimes I)(I\otimes D_h+\frac{\varepsilon-1}{\varepsilon}E_\nabla\otimes I_1))
(\beta_1 I+(1-\beta_1)({W}\otimes I)),
\end{aligned}
\end{equation*}
and the definitions of matrices $B_1$ and $B_2$ are as follows
\begin{equation*}
\begin{aligned}
B_1=-\frac{\lambda}{2}\beta_1(V\otimes D_h),
\quad
B_2=\beta_1(I+\frac{\lambda}{2} V\otimes L_h),
\end{aligned}
\end{equation*}
while the boundary data $\bb{b}_r$ and $\bb{b}_j$ are defined as follows
\begin{align*}
\bb{b}_r 
&=
\frac{\lambda}{2}\beta_2(V^2W\otimes D_h)B^n+
\frac{\lambda}{2}\beta_3(V^2\otimes D_h)B^n+
\frac{\lambda}{2}\bb{f}^*,\\
\bb{b}_j 
&=
-\beta_2
(I+\frac{\lambda}{2} V\otimes L_h)
(VW\otimes I)B^n-
\beta_3(I+\frac{\lambda}{2} V\otimes L_h)
(V\otimes I)B^n
+\frac{\lambda}{2}\bb{g}^*.
\end{align*}
Hence, our goal is to design appropriate algorithms to solve 
the discrete time evolution presented in Eq. (\ref{equ:iter_sys_1}).

\subsection{Discretization scheme for the general case equation}
\label{section:the general case of transport equations}

\par For the general case \eqref{equ:transport} in one space dimension, 
we employ the framework of Jin et al.\cite{Jin2000DiffusiveRelaxation} to decompose this equation into the following relaxation step and convection step.
\begin{itemize}
\item \textit{Relaxation step.}
\begin{equation}\label{equ:7}
    \begin{aligned}
        \partial_{t} r &=-\frac{{\sigma_S}}{\varepsilon^2}(r-\rho), \\
        \partial_{t} j &=-\frac{1}{\varepsilon^2}\left({\sigma_S}j+\left(1-\varepsilon^2\right) v \partial_x r\right).
    \end{aligned}
\end{equation}
\item \textit{Convection step:}
\begin{equation}\label{equ:8}
    \begin{aligned}
        \partial_{t} r+v \partial_x j&={-\sigma_Ar+Q}, \\
        \partial_{t} j+v \partial_x r&={-\sigma_Aj},
    \end{aligned}
\end{equation}
\end{itemize}
where the boundary condition is $j = -{\frac{v}{\sigma}}\partial_x r$.
\par Therefore, we can directly present the discretization scheme of the relaxation-convection scheme. By setting the same $\mathbf{r}^n$, $\mathbf{j}^n$, $\mathbf{r}^*$, and $\mathbf{j}^*$, we obtain the following matrix representation for the relaxation step.
\begin{equation}
\begin{aligned}
\bb{r}^* &= 
(I\otimes \text{diag}
(\beta_1)
\bb{r}^n+
(I\otimes \text{diag}((1-\beta_1)
\bb{\rho}^n,\\
\bb{j}^* &= 
(I\otimes \text{diag}(\beta_1))
\bb{j}^n-
(I\otimes \text{diag}(\beta_2))
(VW\otimes I)
\nabla_{\varepsilon,x}\bb{r}^n
-
(I\otimes \text{diag}(\beta_3))
(V\otimes I)\nabla_{\varepsilon,x}\bb{r}^n\\
&\quad-
(I\otimes \text{diag}(\beta_2)
(VW\otimes I)
B^n-
(I\otimes \text{diag}(\beta_3))(V\otimes I)B^n,
\end{aligned}
\end{equation}
where the parameters are defined as
\begin{equation*}
\begin{aligned}
\beta_1 = e^{-\frac{{\sigma_S(x_m)}\tau}{\varepsilon^2}},\quad
\beta_2 = {\frac{1}{\sigma(x_m)}}(1-e^{-\frac{{\sigma_S}\tau}{\varepsilon^2}}-\frac{{\sigma_S}\tau}{\varepsilon^2}e^{-\frac{{\sigma_S}\tau}{\varepsilon^2}}),\quad
\beta_3 = \frac{\tau}{\varepsilon^2}
e^{-\frac{{\sigma_S}\tau}{\varepsilon^2}}
(1-\varepsilon^2)).
\end{aligned}
\end{equation*}
We maintain the definitions for \( V \), \( W \), \( \nabla_{\varepsilon,x} \), and \( B^n \), where the specific definitions of \( \nabla_{\varepsilon,x}^k \) and \( B_k^n \) are as follows:
\begin{equation*}
\nabla_{\varepsilon,x}^k =
\frac{1}{2h}
\begin{bmatrix}
\frac{-\varepsilon v_k}{\varepsilon v_k+{\sigma h}}    &  1       &                &    &       \\
 -1  &  0        &  1        &    &        \\
 &-1  &  0        &  \ddots        &    &        \\
&     & \ddots    &   \ddots       & \ddots  &   \\
&     &           &    -1      & 0  &    1    \\
 &    &           &                &  -1  &\frac{\varepsilon v_k}{\varepsilon v_k+{\sigma h}} 
\end{bmatrix},\quad
B_k^{n} = 
\begin{bmatrix}
-\frac{{\sigma }}{2(\varepsilon v_k+{\sigma h)}}F_L(v_k)\\
0\\
\vdots\\
0\\
\frac{{\sigma }}{2(\varepsilon v_k+{\sigma h)}}F_R(v_k)
\end{bmatrix}.
\end{equation*}
\par For the discrete matrix formulation of the convection step, we present it directly as follows:
\begin{equation}
\begin{aligned}
 \bb{r}^{n+1} &=  (({1-\tau \sigma_A})I+ \frac{\lambda}{2} (V\otimes I)L_{h}^{\mathrm{bdry}} ) \bb{r}^* - \frac{\lambda}{2}V\otimes D_h\bb{j}^*
+ \frac{\lambda}{2} \bb{f}_v^*+
{\tau\bb{Q}}, \\
\bb{j}^{n+1} &= (({1-\tau\sigma_A})I + \frac{\lambda}{2} V\otimes L_h) \bb{j}^* - \frac{\lambda}{2} (V\otimes I)D_{h}^{\mathrm{bdry}}\bb{r}^* + \frac{\lambda}{2} \bb{g}_v^*.
\end{aligned}
\end{equation}
We maintain the definitions of $\lambda$, $L_{h}^{\mathrm{bdry}}$, and $D_{h}^{\mathrm{bdry}}$, where the specific definitions of $L_{k,h}^{\mathrm{bdry}}$ and $D_{k,h}^{\mathrm{bdry}}$ are as follows:
\begin{equation*}
L_{k,h}^{\mathrm{bdry}} = L_h+\text{diag}\left[\frac{\varepsilon v_k-v_k}{\varepsilon v_k+\sigma h};0;\cdots;0;\frac{\varepsilon v_k-v_k}{\varepsilon v_k+\sigma h}\right],\quad
D_{k,h}^{\mathrm{bdry}} = D_h+\text{diag}\left[\frac{v_k-\varepsilon v_k}{\varepsilon v_k+\sigma h} ;0;\cdots;0;-\frac{v_k-\varepsilon v_k}{\varepsilon v_k+\sigma h}\right].
\end{equation*}
The boundary conditions are also defined in the same way for $\bb{f}^*$ and $\bb{g}^*$, and the specific definitions of $\bb{f}^*_k$ and $\bb{g}^*_k$ are as follows.
     \begin{equation*}
     \bb{f}_{k}^{*} = 
     \begin{bmatrix}
     \frac{{\sigma}hv_k+v_k^2}{\varepsilon v_k+{\sigma}h}F_L(v_k);
     0; \cdots; 0;
     \frac{{\sigma}hv_k+v_k^2}{\varepsilon v_k+{\sigma}h}F_R(v_k)
     \end{bmatrix},\quad
     \bb{g}_{k}^{*} = 
     \begin{bmatrix}
     \frac{{\sigma}hv_k+v_k^2}{\varepsilon v_k+{\sigma}h}F_L(v_k);
     0 ; \cdots ; 0 ;
     -\frac{{\sigma}hv_k+v_k^2}{\varepsilon v_k+{\sigma}h}F_R(v_k)
     \end{bmatrix}.
     \end{equation*}
We do not provide a detailed computation of the overall matrix formulation for the relaxation-convection scheme here.
\subsection{Pre-processing of the discrete evolution}
\label{section:preprocessing}
\par Before formally proceeding with algorithm design, we preprocess the discrete evolution defined in Eq. (\ref{equ:iter_sys_1}). The purpose of preprocessing is to numerically unify the magnitudes of the variables $r(t,x,v)$ and $j(t,x,v)$: in the incoming boundary condition, the relationship between the variables is $j(t,x,v) = -v \partial_x r(t,x,v)$, meaning $j$ is of the same order of magnitude as the partial derivative of $r$ with respect to $x$. This implies that after discretization, $j = \mathcal{O}(\frac{1}{h}r)$. This can also be observed from the numerical magnitudes of the matrices $A_1$, $A_2$, $B_1$, $B_2$, and the boundary data $\mathbf{b}_r$, $\mathbf{b}_j$. First, we analyze the element magnitudes and boundary data in each matrix:
\begin{equation*}
A_1:\,\mathcal{O}(\frac{\tau}{h^2})+o(\frac{\tau}{\varepsilon^2}e^{-\frac{\tau}{\varepsilon^2}}),
\quad 
B_1:\,o(e^{-\frac{\tau}{\varepsilon^2}}),
\quad
A_2:\,\mathcal{O}(\frac{1}{h}+\frac{\tau}{h^2})+o(\frac{\tau}{\varepsilon^2}e^{-\frac{\tau}{\varepsilon^2}}),
B_2:\,\mathcal{O}(e^{-\frac{\tau}{\varepsilon^2}}),
\end{equation*}
while for the element in the boundary vector
\begin{equation*}
\bb{b}_r:\,\mathcal{O}(\frac{\tau}{h^2})+o(\frac{\tau}{\varepsilon^2}e^{-\frac{\tau}{\varepsilon^2}}),
\quad
\bb{b}_j:\,\mathcal{O}(\frac{1}{h}+\frac{\tau}{h^2})+o(\frac{\tau}{\varepsilon^2}e^{-\frac{\tau}{\varepsilon^2}}),
\end{equation*}
Under the CFL condition $\frac{\tau}{h^2} \le 1$ (which will be specified in the theorem), we have $\mathcal{O}\left(\frac{\tau}{h^2}\right) = \mathcal{O}(1)$, and 
$A_1 \sim \mathcal{O}(1), A_2 \sim \mathcal{O}(h^{-1}), B_1 \sim o(1)$. This highlights the importance of preprocessing for the matrix element magnitudes, as it is closely related to the calculation of query complexity.
\par Therefore, to unify the magnitudes and reduce the maximum norm of the coefficient matrix to $\mathcal{O}(1)$, we simply need to divide j by $N_x$. This allows us to preprocess the linear system in Eq. (\ref{equ:iter_sys_1}) through the following variable substitutions. Let $W_d = \text{diag}([w_1,...,w_{N_v}])$, and define the transformed variables: $\hat{\mathbf{r}} = (W_d^{1/2}\otimes I)\mathbf{r}$, $\hat{\mathbf{j}}=\frac{1}{N_x}\mathbf{j}$
with the corresponding matrix transformations as follows:
\begin{equation*}
    \begin{aligned}
        \hat{A}_1=(W_d^{\frac{1}{2}}\otimes I){A}_1(W_d^{-\frac{1}{2}}\otimes I),\quad \hat{B_1}=(W_d^{\frac{1}{2}}\otimes I)(N_x\cdot {B}_1),\quad \hat{A}_2=\frac{A_2}{N_x}(W_d^{-\frac{1}{2}}\otimes I),
    \end{aligned}
\end{equation*}
with the corresponding boundary vector transformations as
\begin{align*}
\hat{\bb{b}}_r=(W_d^{\frac{1}{2}}\otimes I)\bb{b}_r,\quad
\hat{\bb{b}}_j=\frac{1}{N_x}\bb{b}_j.
\end{align*}
Therefore, the actual numerical solution is carried out using the following pre-processed scheme, where at discrete time $n$, the discrete evolution is as follows:
\begin{equation}\label{equ:iter_sys_2}
\begin{aligned}
&\hat{\bb{r}}^{n+1} = 
\hat{A}_1\hat{\bb{r}}^{n}+\hat{B}_1\hat{\bb{j}}^n+\hat{\bb{b}}_r,\\
&\hat{\bb{j}}^{n+1} = 
\hat{A}_2\hat{\bb{r}}^{n}+{B}_2\hat{\bb{j}}^n+\hat{\bb{b}}_j.
\end{aligned}
\end{equation}
After such a transformation, we gain two advantages: the first is that the maximum magnitude of elements in each part of the matrix and the boundary vector is of $\mathcal{O}(1)$. The second is that we performed a similarity transformation $W_d^{1/2}\otimes I$, which handles the structurally simple but highly asymmetric matrix $W\otimes I$, and this is important for the subsequent estimation of query complexity.
\par In the rest calculations, we let variables such as $\bb{r}^n$ and $\bb{j}^n$ represent the post-processed variables $\hat{\bb{r}}^n$ and $\hat{\bb{j}}^n$.

\section{Iterative method based on Schr\"odingerization}
\label{sec:The iterative methods}

\par In this section, we present an algorithm to handle the iterative marching Eq. (\ref{equ:iter_sys_2}) of the pre-processed system. This algorithm is based on the  method proposed by Jin et al.\cite{Jin2024QuantumSO}, which employs the Schr\"odingerization approach. Therefore, we will first introduce the basic procedure of Schr\"odingerization, then present the algorithm within the framework of Jin et al.\cite{Jin2024QuantumSO}, and finally provide a complexity analysis of the algorithm.

\subsection{Schr\"odingerization method for linear ODE}
\label{section:schro}

\par Hamiltonian simulation has stringent requirements on the Hermiticity of the Hamiltonian, which poses a challenge in solving ODEs (or PDEs) using Hamiltonian simulation. The Schr\"dingerization method \cite{Jin2024Schrodingerization,Jin2023Detailes} provides a framework for transforming classical linear ordinary differential equation (ODE) and PDE systems into a form compatible with quantum mechanics. Consider the following constant-coefficient linear ODE system:
\begin{equation}\label{equ:ode1}
\frac{d\mathbf{u}(t)}{dt} = A\mathbf{u}(t) + \mathbf{b},
\end{equation}
where $\mathbf{u}, \mathbf{b} \in \mathbb{C}^d$ and $A \in \mathbb{C}^{d \times d}$ is a time-independent matrix. We homogenize the system by extending the vector $\mathbf{u}$ as follows \cite{Jin2024Schrodingerization}:
\begin{equation}\label{equ:ode2}
\frac{d}{dt}
\begin{bmatrix}
\mathbf{u}(t)\\
\mathbf{b}
\end{bmatrix}
=
\begin{bmatrix}
A & I \\
O & O
\end{bmatrix}
\begin{bmatrix}
\mathbf{u}(t)\\
\mathbf{b}
\end{bmatrix}.
\end{equation}
By adopting the notation $\dot{\mathbf{u}} = A\mathbf{u}$ for the homogenized system \eqref{equ:ode2}, we implicitly set the non-homogeneous term 
$\bb{b}$ to zero in the original system \eqref{equ:ode1}.

\par Next, one  applies the warped transformation $\hat{\mathbf{u}}(t,p) = e^{-p}\mathbf{u}(t)$, initially defined for $p \ge 0$ but naturally extendable to $p < 0$. Leveraging the derivative property of the exponential function, $\partial_t(e^{-p}) = -e^{-p}$, one can transform the system \eqref{equ:ode1} into an equivalent equation for $\hat{\mathbf{u}}$:
\begin{equation}\label{equ:sch_hat_u}
\begin{aligned}
\partial_t \hat{\mathbf{u}}(t,p) &= -A_1 \partial_p
\hat{\mathbf{u}}(t,p) + iA_2 \hat{\mathbf{u}}(t,p), \\
\hat{\mathbf{u}}(0, p) &= e^{-|p|}\mathbf{u}(0),
\end{aligned}  
\end{equation}
where $A_1$ and $A_2$ are the Hermitian and anti-Hermitian parts of $A$, respectively, denoted as:
\begin{equation*}
A_1 = \frac{A + A^{\dagger}}{2} = A_1^{\dagger}, \quad
A_2 = \frac{A - A^{\dagger}}{2i} = A_2^{\dagger}.    
\end{equation*}
The current method's accuracy in $p$ is affected to the limited
regularity of $e^{-|p|}$. According to Jin et al.\cite{Jin2025LinearNonUnitary}, this issue can be improved through various smooth initial data in $p$,  such as cut-off functions, higher-order interpolation, and the Fourier transform, to achieve the near or even optimal complexity. 

\subsubsection{The discrete Schr\"odingerization method}

\par To numerically solve the PDE in Eq. (\ref{equ:sch_hat_u}), we employ the discrete Schrödingerization method here. First, to discretize the $p$ domain, we choose a finite domain $[-L, R]$ with $L, R > 0$ large enough to satisfy $L>\lambda^-_{\max}(A_1) T$ and $R>\lambda^+_{\max}(A_1)T$
where two notations associated with $A_1$ are defined as
\[\lambda_{\max}^-(A_1) :=\sup\{|\lambda|:\lambda\in \sigma(A_1), \lambda<0\},\]
\[\lambda_{\max}^+(A_1) :=\sup\{|\lambda|:\lambda\in \sigma(A_1), \lambda>0\},\]
in which $\delta$ is the desired accuracy. Then, we can impose periodic boundary conditions in the $p$ direction and employ the Fourier spectral method by discretizing the $p$ domain. For this purpose, we define a uniform mesh size $\Delta p = (R + L)/N_p$, where $N_p = 2^{n_p}$ is an even integer. The grid points are given by $-L = p_0 < p_1 < \cdots < p_{N_p} = R$. To compute $\bb{w}(t, p)$, we represent the function values at these grid points as a vector $\hat{\mathbf{u}}_h(t)$. More precisely, we define $\mathbf{u}_h(t) = \sum_{k,i} \hat{\mathbf{u}}_i(t, p_k) \ket{k, i}$, where $\hat{\mathbf{u}}_i$ denotes the $i$-th component of $\hat{\mathbf{u}}$. 

\par By performing a discrete Fourier transform along the $p-$direction, one can obtain
\begin{equation}
\begin{aligned}
\label{heatww}
\frac{\d}{\d t}\mathbf{u}_h(t) &= -\i (P_\mu \otimes  A_1 ) \mathbf{u}_h(t) + \i (I\otimes A_2 ) \mathbf{u}_h(t) ,\\
\mathbf{u}_h(0) &= [
e^{-|p_0|},\cdots,e^{-|p_{N_p-1}|}]^{\top} \otimes \bb{u}_0,
\end{aligned}  
\end{equation}
in which $P_\mu$ represents the matrix form of the momentum operator $-\mathrm{i}\partial_p$. The diagonalization of $P_\mu$ is achieved through the transformation  $P_\mu = \Phi D_\mu \Phi^{-1}$, where $D_\mu = \text{diag}(\mu_0, \dots, \mu_{N_p-1})$ and the Fourier modes are defined as $\mu_k = \frac{2\pi}{R + L} \left( k - \frac{N_p}{2} \right)$. Here, $\Phi$ is the $N_p \times N_p$ matrix representation of the discrete Fourier transform, with elements $\Phi = (\phi_{jl})_{N_p \times N_p} = (\phi_l(x_j))_{N_p \times N_p}$, where $\phi_l(x) = \mathrm{e}^{\mathrm{i} \mu_l (x + L)}$. This diagonalization allows one to map the dynamics back to a Hamiltonian system. Applying the change of variables $\hat{\bb{u}}_h = (\Phi^{-1} \otimes I)\bb{u}_h(t)$, one gets
\begin{equation}\label{generalSchr}
\frac{\d}{\d t} \hat{\bb{u}}_h(t) = -\i H\hat{\bb{u}}_h(t) ,
\end{equation}
where $ H = D_\mu \otimes A_1 -  I \otimes A_2 $.

\subsubsection{Reconstruction of the solution}

\par To reconstruct the state $\bb{u}(t)$ from $\mathbf{w}(t,p)$, one can apply the inverse Fourier transform and project it onto either the positive $p$ domain or at a specific momentum $p = p^*$. Since $A_1$ is Hermitian, it admits $n$ real eigenvalues, ordered as 
$\lambda_1(A_1)\leq \lambda_2(A_1)\leq \cdots\lambda_{n}(A_1)$. A key observation is that if $A_1$ is negative definite (i.e. $\lambda_{n}(A_1)\le 0$), $\bb{u}$ can be recovered from $\bb{w}$ directly at $p = 0$. Additional, from the idea of \cite{Jin2025LinearSystems}, if $A_1$ is not negative definite, the recovery of $\bb{u}$ can instead be achieved using 
\begin{equation*}
\bb{u}(t) = e^{p} \hat{\bb{u}}(t,p),\quad \text{for}\;\text{any}\quad  p\geq p^{\Diamond},
\end{equation*}
where $p^{\Diamond}\geq \lambda_{\max}^+(A_1) t$, 
or by using the integration,
\begin{equation*}
\bb{u}(t) = e^{p}\int_{p}^{\infty} \hat{\bb{u}}(t,q)dq ,\quad \text{for}\;\text{any}\quad  p\geq p^{\Diamond}.
\end{equation*}

\begin{remark}
From the idea of \cite{Jin2025Precondition}, if one smoothens the initial data $\psi(p) \in H^r((-L,R))$, where $r$ is an arbitrary positive integer, then one has
$\Delta p = \mathscr{O}(\sqrt[r]{\delta})$. Therefore, by requiring $(1/\delta)^{1/r} \sim \log(1/\delta)$, one can obtain that
\begin{equation*}
\frac{1}{\Delta p} = \mathcal{O}\left(\log\left(\frac{1}{\delta}\right)\right).
\end{equation*}
\end{remark}

That is, sufficiently smooth initializations can offer nearly exponential speedup in $p$ variable for the Schr\"odingerization method in terms of precision $\delta$.  

\subsection{Configuration for the iterative method}

\subsubsection{Structure of the iterative method}

\par In this section, we will present an algorithm for solving Eq. (\ref{equ:iter_sys_2}) using an iteration method based on Schr\"odingerization. This method employs the fundamental framework proposed by \cite{JL2024}, which is designed to solve general discrete iterative schemes of the form $x_{k+1}=Cx_k$. We let $\bb{x}^n = [\bb{j}^n; \bb{r}^n; 1; 1]$ and transform Eq. (\ref{equ:iter_sys_2}) into the following form:
\begin{equation}\label{equ:x_iter}
\bb{x}^{n+1}:=
\begin{bmatrix}
\bb{j}^{n+1}\\
\bb{r}^{n+1}\\
1\\
1
\end{bmatrix}
=
\begin{bmatrix}
B_2 & A_2 &\bb{b}_j&O\\
B_1 & A_1 &O&\bb{b}_r\\
O&O&1&0\\
O&O&0&1
\end{bmatrix}
\begin{bmatrix}
\bb{j}^{n}\\
\bb{r}^{n}\\
1\\
1
\end{bmatrix}
:=
C\bb{x}^{n}.
\end{equation}
Next, rewrite the iterative format mentioned in Eq.(\ref{equ:x_iter}) as
\begin{equation}\label{equ:dt12}
\bb{x}^{n+1}-\bb{x}^{n} = (C-I)\bb{x}^{n}.
\end{equation}
Further, one can regard $\bb{x}^n$ as $\bb{x}(s)$ ($s$ is a new time variable) and $\bb{x}^{n+1}-\bb{x}^{n}$ as $d\bb{x}(s)/ds$, to obtain the following ODE system:
\begin{equation}\label{equ:dt2}
\frac{d\bb{x}(s)}{ds} = 
(C-I)\bb{x}(s).
\end{equation}
However, the problem we need to solve differs slightly from the one addressed by \cite{JL2024}. They aimed to find the steady-state formulation for an iterative problem, and estimate the evolution time $s$ based on the fidelity error. Here, however, our goal is to solve Eq. (\ref{equ:transport1}) at a specific time $t^{\text{ori}}$, rather than the steady-state solution. To address this issue, we refer to the relationship between the time scale $s$ in Eq. (\ref{equ:dt2}) and the time scale $t$ in the original equation, as discussed in the following Remark \ref{remark:time}, and derive it as $s^{\text{iter}}:=t^{\text{ori}}/\tau=N_t$.
\begin{remark}
\label{remark:time}
\par First, we analyze Eq. (\ref{equ:dt12}). To clarify the temporal relationship, we express $\bb{x}^n$ as $\bb{x}(n\tau)$. For the $n$-th discrete time step, corresponding to the continuous time interval $[n\tau, (n+1)\tau]$, we rewrite Eq. (\ref{equ:dt12}) in the following equivalent form:
\begin{equation*}
\bb{x}((n+1)\tau)-\bb{x}(n\tau)= (C-I)\bb{x}(n\tau).
\end{equation*}
Note that $\tau$ here is determined by the discrete numerical scheme for the diffusion equation. This equation can be further transformed by approximating the difference quotient with a derivative:
\begin{equation}\label{equ:C_I_1}
\frac{d\bb{x}(n\tau)}{dt} \approx 
\frac{1}{\tau}(C-I)\bb{x}(n\tau).
\end{equation}
Introducing the variable substitution $s=\frac{t}{\tau}$ and defining $\bb{x}(s):=\bb{x}(t/\tau)$, we find that equation \eqref{equ:C_I_1} becomes equivalent to Eq. (\ref{equ:dt2}). Consequently, to obtain the solution at original time $t^{\text{ori}}$, we must solve \eqref{equ:dt2} at the scaled time $s^{\text{iter}}=t^{\text{ori}}/\tau=N_t$.
\end{remark}

\subsubsection{Schr\"odingization for the iterative method}

\par To solve  Eq. (\ref{equ:dt2}), we need to convert it into a Hermitian Hamiltonian. Here, we use the Schr\"odingerization method, and the specific steps are as follows. First, we note that $C-I$ is a fixed matrix, so we can directly apply the method mentioned in Section \ref{section:schro}. By using the warped transformation $\hat{\bb{x}}(s,p) = e^{-p}\bb{x}(s)$ and extending it to $p < 0$, then $\hat{\bb{x}}(s,p)$ satisfies that
\begin{equation*}
\begin{aligned}
\partial_{s}\hat{\bb{x}}(s,p)
&= -(C_1-I)\partial_{p}\hat{\bb{x}}(s,p)+iC_2\hat{\bb{x}}(s,p), \\
\hat{\bb{x}}(s=0,p) &= e^{-|p|}\bb{x}_0,
\end{aligned}
\end{equation*}
where we define the Hermitian operators  
$C_1 = (C + C^\dagger)/2 = C_1^\dagger$ and $C_2 = (C - C^\dagger)/(2i) = C_2^\dagger$, which correspond to the real and imaginary parts of $C$, respectively. Consider now the vector $\bb{x}_h(s) = \sum\limits_{k,i} \bb{\hat{x}}_i(s,p_k) \ket{k,i}$, where $\bb{\hat{x}}_i(s,p_k)$ denotes the $i$-th component of the vector $\bb{\hat{x}}(s,p_k)$. Then, we apply the discrete Fourier transform (DFT) along the $p$ direction. This transformation yields 
\begin{equation*}
\begin{aligned}
&\frac{\d}{\d s} \bb{x}_h(s) = -\i (P_\mu \otimes  C_1) \bb{x}_h(s) + \i (I\otimes C_2 ) \bb{x}_h(s) ,\\
&\bb{x}_h(0) = [
e^{-|p_0|},\cdots,e^{-|p_{N_p-1}|}]^{\top} \otimes \bb{x}_0.
\end{aligned}
\end{equation*}
By performing the change of variables $\hat{\mathbf{x}}_h = (\Phi^{-1} \otimes I) \mathbf{\hat{x}}_h$, we obtain
\begin{equation}\label{equ:sch linear_1}
\frac{\d}{\d s} \hat{\bb{x}}_h(s) = -\i H_{\text{total}}^{\text{iter}}\hat{\bb{x}}_h(s) ,
\end{equation}
where $H_{\text{total}}^{\text{iter}}=D_\mu\otimes C_1-I \otimes C_2 $ and $\Phi$ is defined in Section \ref{section:schro}.
\subsection{Query complexity analysis}
\par In quantum computing, query complexity typically refers to the number of calls to a quantum black box (oracle), which differs from the way complexity is measured in classical computing through counting basic operations (such as addition, multiplication) or the number of program steps executed. Berry et al.\cite{Berry2015Hamiltonian}) addressed the Hamiltonian simulation problem by proposing an algorithm based on quantum signal processing and proved its tight bounds on query complexity.
\begin{lemma}\label{lemma:Berry_2015}
\cite{Berry2015Hamiltonian}  
    Given an $\mathcal{S}$-sparse Hamiltonian $H$ acting on $m(H)$ qubits, its simulation within error $\delta$ requires  
    \begin{equation*}  
        \mathcal{O}\left(\chi \log(\chi/\delta) / \log\log(\chi/\delta)\right)  
    \end{equation*}
    queries and
    \begin{equation*}  
        \mathcal{O}\left(  
        \chi\big[m(H) + \log^{2.5}(\chi/\delta)\big]  
        \frac{\log(\chi/\delta)}{\log\log(\chi/\delta)}  
        \right)  
    \end{equation*}  
    additional two-qubit gates, where $\chi(H) = \mathcal{S}(H)\|H\|_{\max} t$. Here, $\mathcal{S}(H)$ denotes the sparsity of $H$, $\|H\|_{\max}$ represents the maximum modulus of the elements in $H$, and $t$ is the evolution time.
\end{lemma}  
\par With the help of Lemma \ref{lemma:Berry_2015}, we can calculate the query complexity of the iterative method based on Schr\"odingerization, i.e., Eq. (\ref{equ:sch linear_1}), where our computational focuses on the maximum norm $\|H\|_{\max}$ and the evolution time $t$.
\begin{theorem}
\label{theorem:The iterative methods}
The Hamiltonian system \eqref{equ:sch linear_1} can be simulated within error $\delta$ with
\begin{equation}
\label{equ:querycomplexity:1}
\mathcal{Q}_{query} = 
\mathcal{\tilde{O}}(N_vN_x^2\log(N_x)),
\quad
\mathcal{Q}_{gate} = 
\mathcal{\tilde{O}}(N_vN_x^2\log(N_x)),
\end{equation}
in which $\mathcal{\tilde{O}}(\cdot)$ represents $\mathcal{O}(\cdot)$ by ignoring the $\log\log$ terms.
\end{theorem}
\begin{proof}
\par According to Lemma \ref{lemma:Berry_2015}, the focus of our calculation lies in the computation of $\chi(H_{\text{total}}^{\text{iter}})$ and $m(H_{\text{total}}^{\text{iter}})$, and the specific steps are as follows:
\begin{enumerate}
\item For the sparsity $\mathcal{S}(H_{\text{total}}^{\text{iter}})$. It is straightforward to verify that it is of $\mathcal{O}(N_v)$, which is introduced by the $W$ in the iterative scheme Eq. (\ref{equ:iter_sys_2}).
\item For the maximum modulus of the elements in $H_{\text{total}}^{\text{iter}}$, i.e., $\|H_{\text{total}}^{\text{iter}}\|_{\max}$. Using the calculation process of Eq. (\ref{equ:sch linear_1}), we can obtain $\|H_{\text{total}}^{\text{iter}}\|_{\max} = \|C\|_{\max}\|D_{\eta}\|_{\max}$, where $C$ is defined in Eq. (\ref{equ:x_iter}) and $\|C\|_{\max} = \mathcal{O}(1)$, and $\|D_{\eta}\|_{\max} = N_{\eta}$. Therefore, we can conclude that $\|H_{\text{total}}^{\text{iter}}\|_{\max} = \mathcal{O}(N_{\eta})$.
\item For the evolution time $s^{\text{iter}}$. Through our explanation in Remark \ref{remark:time} regarding the relationship between the original time $t^{\text{ori}}$ and the time $s^{\text{iter}}$ based on the Schrödingerized iterative method, we can determine that its value is of $\mathcal{O}(N_t)$.
\item For the matrix size of $H_{\text{total}}^{\text{iter}}$. It is evident that we can determine the matrix dimensions to be $N_\eta N_v N_x$.
\end{enumerate}
Therefore, through our calculations, we can obtain the values of $\chi(H_{\text{total}}^{\text{iter}})$ and the number of gates $m(H_{\text{total}}^{\text{iter}})$:
\begin{equation*}
\begin{aligned}
\chi(H_{\text{total}}^{\text{iter}}) &= \mathcal{O}(N_vN_tN_{\eta}\log(1/\delta)) =
\mathcal{O}(N_vN_x^2\log(N_x)),\\
m(H_{\text{total}}^{\text{iter}}) &= \log(N_{\eta}N_vN_x)
= \mathcal{O}(\log(N_{\eta})+\log(N_v)+\log(N_x)),
\end{aligned}
\end{equation*}
where we set $\delta = \mathcal{O}(1/N_x)$. For $N_{\eta}$, Ma et al.\cite{Jin2025LinearNonUnitary} demonstrates that by using a smooth initialization of the warped phase transform, the Schr\"odingerization method can achieve optimal scaling in matrix queries-specifically, $N_{\eta} = \mathcal{O}(\log(1/\delta)) = \mathcal{O}(\log(N_x))$.
\par Finally, by substituting the values of $\chi(H_{\text{total}}^{\text{iter}})$, $m(H_{\text{total}}^{\text{iter}})$, and $N_\eta$ into Lemma \ref{lemma:Berry_2015}, we can obtain the query complexity and gate complexity of using the simulated Hamiltonian $A$ as follows
\begin{equation*}
\mathcal{Q}_{query} = 
\mathcal{\tilde{O}}(N_vN_x^2\log(N_x)),
\quad
\mathcal{Q}_{gate} = 
\mathcal{\tilde{O}}(N_vN_x^2\log(N_x)).
\end{equation*}
\end{proof}
\par 
We compare our query complexity with that of \cite{Xiaoyang2023TimeCA}, which achieves a quantum query complexity of $\mathcal{O}(N_v^2N_x^2\log(N_x))$, and it is not difficult to see that one order of $N_v$ is reduced, which benefits from the pre-process step we designed in Eq. (\ref{equ:iter_sys_2}).
\section{Steady-state solution method based on Schr\"odingerization}
\label{sec:Steady-state solution}
\par In Section \ref{sec:The iterative methods}, we presented an iterative solution method for Eq. (\ref{equ:iter_sys_2}) based on the Schr\"odingerization-based iterative approach proposed by Jin et al.\cite{Jin2024QuantumSO}, and demonstrated that the query complexity of this iterative method is superior to that based on the HHL algorithm\cite{Xiaoyang2023TimeCA}. However, the iterative method can only solve the state at a specific time and cannot explicitly represent all intermediate states. This limitation makes it difficult for handling time-dependent or physical boundary/interface problems. Therefore, in this section, we propose a method called the steady-state solution method based on Schr\"odingerization, which can effectively address these issues.

\subsection{Construction for the steady-state solution method}

\par First, we present the construction  of the method, which is based on the approach proposed by Jin et al.\cite{Jin2022TimeCA}, transforming the problem into a linear algebraic problem $H\bb{y}=\bb{F}$. This steady state solution method also represents an improvement over the algorithm proposed by He et al.\cite{Xiaoyang2023TimeCA} for solving multiscale linear transport equations, specifically by replacing their choice of the HHL algorithm for solving linear equations with a Hamiltonian simulation method.

\subsubsection{Structure of the steady-state solution method}

\par We  first
concatenating all time steps of $\bb{r}^n$ to obtain $S_1=[\bb{r}^{N_t};\bb{r}^{N_t-1};\cdots;\bb{r}^1]$, and concatenating all time steps of $\bb{j}^n$ to obtain $S_2=[\bb{j}^{N_t};\bb{j}^{N_t-1};\cdots;\bb{j}^1]$, and setting $\bb{y}=[S_2;S_1]$. Then, one can represent the iterative scheme in Eq. (\ref{equ:iter_sys_2}) as a linear system, where $H$ and $\bb{F}$ are defined as
\begin{equation}\label{equ:y_iter}
H:= 
\begin{bmatrix}
L_{11} & L_{12} \\
L_{21} & L_{22}
 \end{bmatrix},\quad
\bb{F}=\begin{bmatrix}
\bb{F}_1\\
\bb{F}_2
\end{bmatrix},
\end{equation}
in which the block matrices are defined as follows
\begin{equation*}
    L_{11} =
    \begin{bmatrix}
    I    & -B_2 &          &          \\
              &  I   &   \ddots &          \\
              &           &   \ddots & -B_2 \\
              &           &          & I   \\
    \end{bmatrix},
    L_{12} =
    \begin{bmatrix}
    O    &  -A_2&           &          \\
              &  O   &  \ddots   &          \\
              &           &   \ddots  & -A_2 \\
              &           &           & O   \\
    \end{bmatrix},
    \bb{F}_{1} =
    \begin{bmatrix}
    \bb{b}_r\\
    \bb{b}_r\\
    \vdots\\
    \bb{b}_r+A_2\bb{r}^0+B_2\bb{j}^0\\
    \end{bmatrix},
\end{equation*}
\begin{equation*}
    L_{21} =
    \begin{bmatrix}
    O    &  -B_1&           &          \\
              &  O   &  \ddots   &          \\
              &           &   \ddots  & -B_1 \\
              &           &           & O   \\
    \end{bmatrix},
    L_{22} =
    \begin{bmatrix}
    I    & -A_1 &          &          \\
              &  I   &   \ddots &          \\
              &           &   \ddots & -A_1 \\
              &           &          & I   \\
    \end{bmatrix},
    \bb{F}_{2} =
    \begin{bmatrix}
    \bb{b}_j\\
    \bb{b}_j\\
    \vdots\\
    \bb{b}_j+\hat{A}_1\bb{r}^0+\hat{B}_1\bb{j}^0\\
    \end{bmatrix}.
\end{equation*}
Note that $B_1=o(e^{-\frac{\tau}{\varepsilon^2}})$, so $H$ can be approximately considered as an upper triangular matrix.
\par In fact, we only need to make a certain order adjustment to the variables involved in Eq. (\ref{equ:y_iter}) to make the theoretical analysis of stability very straightforward. We can use the approach of \cite{Hu2024QuantumMultiscale} and consider Eq. (\ref{equ:iter_sys_2}) as a two-variable system. Let $\bb{y}^n=[\bb{r}^n;\bb{j}^n]$, and we  obtain
\begin{equation*}
\bb{y}^{n+1}:=
\begin{bmatrix}
\bb{r}^{n+1}\\
\bb{j}^{n+1}
\end{bmatrix}
=
\begin{bmatrix}
A_1 & B_1\\
A_1 & B_2
\end{bmatrix}
\begin{bmatrix}
\bb{r}^{n}\\
\bb{j}^{n}
\end{bmatrix}
+
\begin{bmatrix}
\bb{b}_r\\
\bb{b}_j
\end{bmatrix}
:=
A\bb{y}^{n}+\bb{b}.
\end{equation*}
Furthermore, we let $\bb{y}=[\bb{y}^{N_t};\bb{y}^{N_t-1};\cdots;\bb{y}^1]$, and define the matrices $H$ and $\mathbf{F}$ as follows
\begin{equation}\label{equ:y_iter_new}
H:=\begin{bmatrix}
    I   & A &          &      \\
        &  I   &   \ddots &      \\
        &      &   \ddots & A \\
        &      &          & I    \\
    \end{bmatrix},\quad
    \bb{F} =
    \begin{bmatrix}
    \bb{b}\\
    \bb{b}\\
    \vdots\\
    \bb{b}+A\bb{y}^0\\
    \end{bmatrix}.
\end{equation}
Note that $H$ is an upper triangular matrix with all diagonal elements equal to 1.
\par It is evident that these two methods only differ in the arrangement of variables, which consequently modifies the structure of $H$ and $\bb{F}$. Fundamentally, however, their stability properties, i.e. the existence and uniqueness of stable solutions, are completely identical. Given that the stability of the second method is more straightforward to analyze, we will henceforth use its stability characterization to represent that of the first method. Thus, regardless of the choice of $\bb{y}$ or the specific construction of $H$ and $\bb{F}$, we can always treat $H$ as an approximation of an upper triangular matrix and derive the following linear system:
\begin{equation}\label{equ:linear system1}
H\bb{y}=\bb{F}.
\end{equation}
To solve the linear system defined by Eq. (\ref{equ:linear system1}), we can use the autonomous system method proposed by Cao et al.\cite{Cao2023QuantumSF}, transforming it into the following linear equations problem:
\begin{equation}\label{equ:linear system2}
\frac{\d \bb{y}(T)}{dT}=\bb{F}-H\bb{y}(T),
\end{equation}
in which $T$ is the newly introduced time in order to distinguish the differences at different times. First, we need to emphasize that our method is feasible. This is because the eigenvalues of $-H$ are all $-1$, which ensures that the ODE in Eq. (\ref{equ:linear system2}) must have a steady-state solution, and that this steady-state solution is unique. This is also an important prerequisite for the applicability of the steady-state solution method. Next, we also need to explain the advantages of the two structures shown in Eq. (\ref{equ:linear system1}) compared to existing methods. According to the conclusions of He et al.\cite{Xiaoyang2023TimeCA}, we can observe that the matrix $H$ of the linear system designed in \cite{Xiaoyang2023TimeCA} is not an upper triangular matrix. As a result, theoretically calculating its eigenvalues becomes challenging, which means stability cannot be guaranteed, and stable solutions may not obviously exist.

\subsubsection{Schr\"odingization for the steady-state solution method}

\par In the following, we focus on studying the structure of the linear system (\ref{equ:y_iter}) and solving the ODE (\ref{equ:linear system2}). To handle this ODE using Hamiltonian simulation, we need to employ the Schr\"odingerization method mentioned in Section \ref{section:schro}. First, we need to eliminate the non-homogeneous term $\bb{F}$ in Eq.(\ref{equ:linear system2}) as done in \cite{Jin2025LinearSystems}. Let $\tilde{\bb{y}}(t)=[\bb{y}(t);1]$ with $\tilde{\bb{y}}(0) = \tilde{\bb{y}}_0$, one obtains the following equation
\begin{equation*}
\begin{aligned}
\frac{\d\tilde{\bb{y}}(T)}{\d T}&=
\begin{bmatrix}
    -H & \bb{F}\\
    O & 1
\end{bmatrix}
\tilde{\bb{y}}(T).
\end{aligned}
\end{equation*}
One can then directly apply the Schr\"odingization method. First, decompose the Hamiltonian $H$ into a Hermitian component $H_1 = (H + H^\dagger)/2$ (with $H_1^\dagger = H_1$) and an anti-Hermitian component $H_2 = (H - H^\dagger)/(2i)$ (with $H_2^\dagger = H_2$). By introducing the warped transformation $\hat{\bb{y}}(t,p) = e^{-p}\tilde{\bb{y}}(t)$ and extending it to $p < 0$, one finds that $\hat{\bb{y}}(t,p)$ satisfies
\begin{equation*}
\begin{aligned}
&\partial_{T}\hat{\bb{y}}(T,p)
= -H_1\partial_{p}\hat{\bb{y}}(T,p)+iH_2\hat{\bb{y}}(T,p), \\
&\hat{\bb{y}}(T=0,p) = e^{-|p|}\tilde{\bb{y}}_0.
\end{aligned}
\end{equation*}
Let $\bb{y}_h(T) = \sum\limits_{k,i} \bb{\hat{y}}_i(T,p_k) \ket{k,i}$, where $\bb{y}_i$ denotes the $i$-th component of $\bb{\hat{y}}$. Applying the discrete Fourier transform (DFT) along the $p-$direction, one gets
\begin{equation*}
\begin{aligned}
&\frac{\d}{\d T} \bb{y}_h(T) = -\i (P_\mu \otimes  H_1 ) \bb{y}_h(T) + \i (I\otimes H_2 ) \bb{y}_h(T),\\
&\bb{y}_h(0) = [
e^{-|p_0|},\cdots,e^{-|p_{N_p-1}|}]^{\top} \otimes \hat{\bb{y}}_0.
\end{aligned}
\end{equation*}
By preforming the change of variables $\hat{\bb{y}}_h(T)= (\Phi^{-1} \otimes I)\bb{y}_h(T)$, one obtains
\begin{equation}\label{equ:sch linear_2}
\frac{\d}{\d T} \tilde{\bb{u}}_h(T) = -\i H_{\text{total}}^{\text{stea}}\tilde{\bb{u}}_h(T) ,
\end{equation}
where $H_{\text{total}}^{\text{stea}}= D_\mu \otimes H_1 -  I \otimes H_2 $ and $\Phi$ is defined in Section \ref{section:schro}.
\subsection{Query complexity analysis}
\subsubsection{Methods for estimating the global error in steady-state solutions}
\label{section:error}
\par Through Lemma \ref{lemma:Berry_2015}, sparsity, the maximum modulus of the elements, and the evolution time are the three most critical factors. In the steady-state solution method, because our discretization scheme includes an explicit part, it is necessary to satisfy the CFL condition $\tau/h^2<1$. This weakens the influence of the maximum modulus of the elements, making the primary task the computation of the evolution time. Additionally, since the new time $T$ we introduced is independent of the original time, we can only calculate it based on the time $T^{\text{stea}}$ when the global relative error is less than $\delta$, i.e.
\begin{equation}
    \begin{aligned}
        \label{equ:relation:error}
        \frac{\hat{\bb{y}}_h(T^{\text{stea}})-(\hat{\bb{y}}_h)_{\infty}}{\hat{\bb{y}}_h(T)-(\hat{\bb{y}}_h)_{\infty}}<\delta,
    \end{aligned}
\end{equation}
where $(\hat{\bb{y}}_h)_{\infty}$ is the steady solution of $\hat{\bb{y}}_h(T)$. Before specifically computing the time, we present an error analysis method for linear ODEs.
\par In fact, to calculate the evolution time through error analysis, we need to consider the system Eq.(\ref{equ:sch linear_2}) after Schr\"odingerization. However, computing this system is overly complex. Here, we cite the conclusion from Hu et al. \cite{Hu2024QuantumMultiscale}, which provides the relationship between the errors of system Eq. (\ref{equ:sch linear_2}) and Eq. (\ref{equ:linear system2}):
\begin{lemma}
\label{lemma:error:1}
\par The induced norms of $\hat{\bb{y}}_h(T)-(\hat{\bb{y}}_h)_{\infty}$ and $\bb{y}(T)-\bb{y}_{\infty}$ satisfy
\begin{equation*}
    \begin{aligned}
        \|\hat{\bb{y}}_h(T)-(\hat{\bb{y}}_h)_{\infty}\|_2=C\|\bb{y}(T)-\bb{y}_{\infty}\|_2,
    \end{aligned}
\end{equation*}
in which $C\le1$ is an constant independent of time $t$.
\qed
\end{lemma}
\noindent Therefore, the time calculated through $(\bb{y}(T^{\text{stea}})-(\bb{y})_{\infty})/(\bb{y}(T)-(\bb{y})_{\infty})<\delta$ can ensure that Eq. (\ref{equ:relation:error}) is satisfied.
\par Through the definition of $\bb{y}(T)$ in Eq. (\ref{equ:linear system2}), we can obtain the following estimation method for relative error, which is achieved by considering the 2-norm of the matrix exponential. Compared to existing methods, it is tighter and thus yields better results.
\begin{lemma}
\label{lemma:error:2}
\par For linear equation presented in Eq. (\ref{equ:linear system2}), $\bb{y}(T)$ converges to the steady state $\bb{y}_{\infty}$ as
\begin{equation}
    \begin{aligned}
        \label{equ:error:1}
        \| \bb{y}(T)-\bb{y}_{\infty}\|_2\le \| e^{-HT}\|_2\cdot\|\bb{y}_0-\bb{y}_{\infty}\|)_2.
    \end{aligned}
\end{equation}
\begin{proof}
\par The analytical solution to the ODE in Eq.(\ref{equ:linear system2}) is
\begin{equation*}
    \begin{aligned}
        \bb{y}(T)-\bb{y}_{\infty}=e^{-HT}(\bb{y}_0-\bb{y}_{\infty}).
    \end{aligned}
\end{equation*}
Consequently, applying the $2-$norm to both sides and leveraging the norm consistency property yields:
\begin{equation*}
    \begin{aligned}
        \|\bb{y}(T)-\bb{y}_{\infty}\|_2\le \| e^{-HT}\|_2\cdot\| \bb{y}_0-\bb{y}_{\infty}\|_2.
    \end{aligned}
\end{equation*}
This completes the proof.
\end{proof}
\end{lemma}
\noindent Therefore, through Lemma \ref{lemma:error:1} and \ref{lemma:error:2}, we can calculate the evolution time $T^{\text{stea}}$ by determining when $\| e^{-HT^{\text{stea}}} \|_2 < \delta$.
\subsubsection{Upper bound for\texorpdfstring{ $\| e^{-HT}\|_2$ }{}via the Laplace transform}
\label{section:upperbound}
\par In the previous section, we presented a method for error control using $\|e^{-HT}\|_2$. If we simply bound the norm of this matrix exponential, the computed result will be very poor. Therefore, directly calculating $e^{-HT}$ and bounding the 2-norm of this matrix will make the inequality tighter. However, computing this matrix exponential using its definition is quite complex. Here, we employ a method based on the Laplace transform and its inverse\cite{Hu2024FundamentalPO,Hu2024TransferPathways,Hu2024Keymotifs} to compute this matrix exponential, which relies on the following property.
\begin{remark}
\par For a matrix $A$, the Laplace transform of the matrix exponential $e^{At}$ is given by:
\begin{equation*}
    \begin{aligned}
        \mathcal{L}[e^{At}](s)=(sI-A)^{-1}.
    \end{aligned}
\end{equation*}
Note that we impose no requirements on $A$, such as Hermiticity, diagonalizability, etc.
\qed
\end{remark}
\par Through the Laplace transform, we convert the computation of the matrix exponential $e^{-HT}$ into solving for the inverse matrix. Based on the matrix $H$ defined in Eq. (\ref{equ:y_iter}), it is a near-upper triangular matrix but not strictly upper triangular, which makes solving for the inverse matrix highly complex. Therefore, here we choose to decompose $H$ into the sum of two matrices, i.e., $H=\bar{H}+E$, where $E$ consists of all elements of $H$ that contain the term $\mathcal{O}(\frac{\tau}{\varepsilon^2}e^{-\frac{\tau}{\varepsilon^2}})$ (excluding other terms containing $\varepsilon$), and $\bar{H}$ represents the remaining part, defined as $\bar{H}=\begin{bmatrix} 
\bar{L}_{11} & \bar{L}_{12} \\ 
\bar{L}_{21} & \bar{L}_{22} 
\end{bmatrix}$. The specific block matrices are simplified as follows:
\begin{equation*}
    \bar{L}_{11} = I,\quad
    \bar{L}_{12} =
    \begin{bmatrix}
    O    &  -\bar{A}_2&           &          \\
              &  O   &  \ddots   &          \\
              &           &   \ddots  & -\bar{A}_2 \\
              &           &           & O   \\
    \end{bmatrix},\quad
    \bar{L}_{21} = O,\quad
    \bar{L}_{22} =
    \begin{bmatrix}
    I    & -\bar{A}_1 &          &          \\
              &  I   &   \ddots &          \\
              &           &   \ddots & -\bar{A}_1 \\
              &           &          & I   \\
    \end{bmatrix},
\end{equation*}
where the specific definitions of $\bar{A}_1$ and $\bar{A}_2$ are as follows:
\begin{align*}
\bar{A}_1 &= (W_d^{\frac{1}{2}}\otimes I)\left[
(I+\frac{\lambda}{2} (V\otimes I)(I\otimes L_h+\frac{\varepsilon-1}{\varepsilon}E_\nabla\otimes I_2))({W}\otimes I)\right.\\
&\quad\left.+\frac{\lambda}{4h}
(V\otimes D_h)({VW}\otimes I)
(I\otimes D_h+E_\nabla\otimes I_1)
\right](W_d^{-\frac{1}{2}}\otimes I),\\
\bar{A}_2 &= 
-\frac{1}{N_x}
\left[ 
(I+\frac{\lambda}{4h} V\otimes L_h)
(V\otimes I)(W\otimes I)
(I\otimes D_h+E_\nabla\otimes I_1)\right.\\
&\quad\left.+\frac{\lambda}{2}
(V\otimes I)(I\otimes D_h+\frac{\varepsilon-1}{\varepsilon}E_\nabla\otimes I_1)(W\otimes I)
\right]
(W_d^{-\frac{1}{2}}\otimes I).
\end{align*}
\par Based on the decomposition of $H$, we can transform the computation of $e^{-HT}$ into the computation of its principal term $e^{\hat{H}T}$, while the influence of $E$ is calculated in the Appendix Section \ref{section:error:analysis}. Its impact on the 2-norm magnitude is of $\mathcal{O}(\frac{\tau}{\varepsilon^2}e^{-\frac{\tau}{\varepsilon^2}})$, which can therefore be ignored when $\varepsilon=o(h/\log h^{-1})$. Next, we use the Laplace transform and its inverse transform to solve for $e^{\hat{H}T}$. Since $\bar{H}$ is an upper triangular matrix, its block inverse matrix $(sI-\bar{H})^{-1}$ can be easily solved, and its specific structure is as follows:
\begin{equation*}
\begin{aligned}
(sI+\bar{H})^{-1} = 
\left[\begin{array}{cccc:cccc}
\frac{1}{s+1}I&&&&O & \frac{\bar{A}_2}{(s+1)^2} & \cdots & \frac{\bar{A}_2\bar{A}_{1}^{N-2}}{(s+1)^N}\\
&\frac{1}{s+1}I&&&&O & \cdots & \frac{\bar{A}_2\bar{A}_1^{N-3}}{(s+1)^{N-1}}\\
&&\ddots&&&&\ddots & \vdots\\
&&&\frac{1}{s+1}I&&&&O\\
\hdashline
&&&&\frac{1}{s+1}I & \frac{\bar{A}_1}{(s+1)^2} & \cdots & \frac{\bar{A}_1^{N-1}}{(s+1)^N}\\
&&&&&\frac{1}{s+1}I & \cdots & \frac{\bar{A}_1^{N-2}}{(s+1)^{N-1}}\\
&&&&&&\ddots & \vdots\\
&&&&&&&\frac{1}{s+1}I
 \end{array}\right].
\end{aligned}
\end{equation*}
Therefore, we can use the inverse Laplace transform and apply its property that $(s+1)^{-(m+1)}A^m$ is $\frac{t^{m}}{m!}e^{-t}A^m$, thus obtaining the specific expression for $e^{-\bar{H}t}$ as:
\begin{equation}\label{equ:exp_H}
\begin{aligned}
e^{-\bar{H}t} = e^{-t}
\left[\begin{array}{cccc:cccc}
I&&&&O & \bar{A}_2t & \cdots & \bar{A}_2\hat{A}_{1}^{N-2}\frac{t^{N-1}}{(N-1)!}\\
&I&&&&O & \cdots & \bar{A}_2\bar{A}_1^{N-3}\frac{t^{N-2}}{(N-2)!}\\
&&\ddots&&&&\ddots & \vdots\\
&&&I&&&&O\\
\hdashline
&&&&I & \bar{A}_1t & \cdots & \bar{A}_1^{N-1}\frac{t^{N-1}}{(N-1)!}\\
&&&&&I & \cdots & \bar{A}_1^{N-2}\frac{t^{N-2}}{(N-2)!}\\
&&&&&&\ddots & \vdots\\
&&&&&&&I
 \end{array}\right].
 \end{aligned}
\end{equation}
\par We have provided the specific structure of $e^{-\bar{H}T}$ in Eq. (\ref{equ:exp_H}), and we can observe a certain pattern: each of the four blocks is an upper triangular matrix, and the matrix blocks along each sub-diagonal within the sub-block matrices are identical. Based on this and combined with the Appendix Section \ref{section:error:analysis}, we can derive the following upper bound estimate for $\|e^{-\bar{H}T}\|_2$. Since our calculations involve very few approximations, the upper bound obtained through this method is actually very close to the true value of $\|e^{-\bar{H}T}\|_2$:
\begin{equation}\label{equ:exp_H_esi}
\begin{aligned}
\|e^{-HT}\|_2 
&\le e^{-T}
\left(1+\sum_{k=0}^{N-1}\frac{\|\bar{A}_1\|_2^kT^k}{k!}+\|\bar{A}_2\|_2\sum_{k=1}^{N-1}\frac{\|\hat{A}_1\|_2^{k-1}T^k}{k!}\right)\\
&\le e^{-T}\left(1+
e^{\|\bar{A}_1\|_2T}
+\frac{\|\bar{A}_2\|_2}{\|\bar{A}_1\|_2}
e^{\|\bar{A}_1\|_2T}
\right)\\
&\le \left(2+\frac{\|\bar{A}_2\|_2}{\|\bar{A}_1\|_2}\right)e^{(-1+\|\bar{A}_1\|_2)T}.
\end{aligned}
\end{equation}
Based on this, we transform the calculation of evolution time into the estimation of the upper and lower bounds of $\|\bar{A}_1\|_2$ and the upper bound of $\|\bar{A}\|_2$, with their specific values calculated in the Appendix Sections \ref{section:A1:lower}, \ref{section:A1:upper} and \ref{section:A2:upper}.

\subsubsection{Estimation for the evolutionary time and the query complexity}
\par In this section, we will first provide an estimate of the evolution time $T^{\text{stea}}$, and then based on this, give an estimate of the query complexity. In Section \ref{section:error}, we mentioned using $\|e^{-HT}\|_2 < \delta$ to estimate the time. Subsequently, in Section \ref{section:upperbound}, we decomposed $H = \bar{H} + E$ and used the conclusions from the Appendix Section \ref{section:error:analysis} to show that the impact of $E$ on the matrix exponential is $\mathcal{O}\left(\frac{\tau}{\varepsilon^2}e^{-\frac{\tau}{\varepsilon^2}}\right)$ (which can be neglected under the condition $\varepsilon = o(h)$), thus allowing us to use $\|e^{-\bar{H}T}\|_2$ instead of $\|e^{-HT}\|_2$. Finally, we provided the specific process for $e^{-\bar{H}T}$ and used the upper bound for $\|e^{-\bar{H}T}\|_2$ given by Eq. (\ref{equ:exp_H_esi}) to estimate the time, i.e. we can estimate $T^{\text{stea}}$ by solving the following inequality
\begin{align*}
\left(2+\frac{\|\bar{A}_2\|_2}{\|\bar{A}_1\|_2}\right)e^{(-1+\|\bar{A}_1\|_2)T^{\text{stea}}}
\le \delta.
\end{align*}
Under the condition that $\tau/h^2 < \frac{10}{11}$, we can use the conclusions from the Appendix Sections \ref{section:A1:lower}, \ref{section:A1:upper} and \ref{section:A2:upper}, namely $1-\|\bar{A}_1\|_2\gtrsim N_t^{-1}$, $\|\bar{A}_1\|_2\gtrsim N_v^{-1}$, and $\|\bar{A}_2\|\lesssim N_v^{\frac{1}{2}}$, to obtain the lower bound estimate of $T^{\text{stea}}$ as follows: $T^{\text{stea}}$ as follows
\begin{equation*}
\begin{aligned}
T^{\text{stea}}
&>\frac{-\log\frac{\delta}{2}+\log\left(2+\frac{\|\bar{A}_2\|_2}{\|\bar{A}_1\|_2}\right)}{1-\|\bar{A}_1\|_2}
&\succsim N_t\log(\frac{1}{\delta}).
\end{aligned}
\end{equation*}
\par Based on the calculation of the evolution time $T^{\text{stea}}$, we can estimate the query complexity of the steady-state solution method based on Schr\"odingerization through Lemma \ref{lemma:Berry_2015} as follows:
\begin{theorem}
\label{theorem:Steady-state solution}
The Hamiltonian dynamics described by \eqref{equ:sch linear_2} can be numerically approximated with precision $\delta$ using
\begin{equation*}
\mathcal{Q}_{query} = 
\mathcal{\tilde{O}}(N_vN_x^2\log(N_x)),
\quad
\mathcal{Q}_{gate} = 
\mathcal{\tilde{O}}(N_vN_x^2\log(N_x)),
\end{equation*}
where $\mathcal{\tilde{O}}(\cdot)$ denotes $\mathcal{O}(\cdot)$ while omitting logarithmic factors of lower order.
\end{theorem}
\begin{proof}
\par Following Lemma \ref{lemma:Berry_2015}, our analysis primarily concerns the evaluation of $\chi(H_{\text{total}}^{\text{stea}})$ and $m(H_{\text{total}}^{\text{stea}})$, with the detailed procedure outlined below:
\begin{enumerate}
\item For the sparsity parameter $\mathcal{S}(H_{\text{total}}^{\text{stea}})$. Direct examination reveals this quantity scales as $\mathcal{O}(N_v)$, originating from the operator $W$ in the iterative formulation Eq. (\ref{equ:iter_sys_2}).
\item For the maximal element magnitude $\|H_{\text{total}}^{\text{stea}}\|_{\max}$. Through analysis of Eq. (\ref{equ:sch linear_2}), we establish $\|H_{\text{total}}^{\text{stea}}\|_{\max} = \|H_1\|_{\max}\|D_{\eta}\|_{\max}$, with $\|H_1\|_{\max} = \mathcal{O}(1)$ and $\|D_{\eta}\|_{\max} = N_{\eta}$. Consequently, $\|H_{\text{total}}^{\text{stea}}\|_{\max} = \mathcal{O}(N_{\eta})$.
\item For the evolution time $T^{\text{stea}}$. Our previous derivation yields $T^{\text{stea}} = \mathcal{O}(N_t\log(1/\delta))$.
\item For the matrix dimensionality of $H_{\text{total}}^{\text{stea}}$. The system size is clearly given by $N_\eta N_v N_x N_t$.
\end{enumerate}
These computations enable us to determine $\chi(H_{\text{total}}^{\text{stea}})$ and the gate count $m(H_{\text{total}}^{\text{stea}})$:
\begin{equation*}
\begin{aligned}
\chi(H_{\text{total}}^{\text{stea}}) &= \mathcal{O}(N_vN_tN_{\eta}\log(1/\delta)) =
\mathcal{O}(N_vN_x^2\log(N_x)),\\
m(H_{\text{total}}^{\text{stea}}) &= \log(N_{\eta}N_vN_xN_t)
= \mathcal{O}(\log(N_{\eta})+\log(N_v)+\log(N_x)+\log(N_t)).
\end{aligned}
\end{equation*}
Here we take $\delta = \mathcal{O}(1/N_x)$. The parameter $N_{\eta}$ satisfies $N_{\eta} = \mathcal{O}(\log(1/\delta)) = \mathcal{O}(\log(N_x))$ as established in \cite{Jin2025LinearNonUnitary}.
\par By inserting the derived expressions for $\chi(H_{\text{total}}^{\text{stea}})$, $m(H_{\text{total}}^{\text{stea}})$, and $N_\eta$ into Lemma \ref{lemma:Berry_2015}, we arrive at the following computational complexities for simulating Hamiltonian $A$:
\begin{equation*}
\mathcal{Q}_{query} = 
\mathcal{\tilde{O}}(N_vN_x^2\log(N_x)),
\quad
\mathcal{Q}_{gate} = 
\mathcal{\tilde{O}}(N_vN_x^2\log(N_x)).
\end{equation*}
\end{proof}
It can be seen that query complexity for the steady-state solution method based on Schr\"odingerization is the same as the iterative method based on Schr\"odingerization, but this approach can give numerical solutions at all intermediate states, thus are of important value in applications.

\section{Numerical examples}
\label{sec:Numerical examples}

\par In this section, we validate our two algorithms  through numerical simulations by varying physical parameters and data.  Here, we solve the  one-dimensional problems  with the following parameters:
\begin{itemize}
    \item \textit{Problem I.} $x\in[0,1]$, $F_L(v)=1$, $F_R(v)=0$, $\sigma_S=1$, $\sigma_A=0$, $Q=0$, $\varepsilon\in[10^{-1}, 10^{-8}]$.
    \item \textit{Problem II.} $x\in[0,1]$, $F_L(v)=0$, $F_R(v)=0$, $\sigma_S=1+(10x)^2$, $\sigma_A=0$, $Q=1$, $\varepsilon\in[10^{-1}, 10^{-8}]$.
    \item \textit{Problem III.} $x\in[0,1]$, $F_L(v)=v$, $F_R(v)=0$, $\sigma_S=1$, $\sigma_A=0$, $Q=1$, $\varepsilon\in[10^{-1}, 10^{-8}]$.
\end{itemize}
In the steady-state solution method, zero initial values are used for all variables and the evolution time is $2N_t$. In the subprogram related to Schr\"odingerization, we set $L=R=N_x$ with using the single-point recovery method. For the variable $N_v$, the $S_8$ standard Gaussian quadrature set on the interval $[-1,1]$ is employed. 
\subsection{Iterative method based on Schr\"odingerization}
\par For the iterative method based on Schr\"odingerization, we adopt the following approach for algorithmic comparison. We directly solve the iterative equation in Eq. (\ref{equ:x_iter}) to represent the solution to the  original problem for comparison, whose results are plotted by dotted lines. We also employ the Schr\"odingerization method for numerical solution, marked as black dots. By examining the results in Fig. \ref{fig:result:1}, Fig. \ref{fig:result:3} and Fig. \ref{fig:result:5}, it can be observed that the numerical results obtained using Schr\"odingerization are identical to those from directly solving the iterative equation. This demonstrates that our iterative method based on Schr\"odingerization is capable of effectively solving this problem and obtain good approximations. 

\subsection{Steady-state solution method based on Schr\"odingerization}

\par For the steady-state solution method based on Schr\"odingerization, we use the solution of the linear system defined in Eq. (\ref{equ:linear system1}) as the original result for comparison, represented by a dotted line, while the numerical solution obtained via Schr\"odingerization is indicated by black dots. The results shown in Fig. \ref{fig:result:2}, Fig. \ref{fig:result:4} and Fig. \ref{fig:result:6} also demonstrate that our steady-state solution method can effectively compute these two numerical problems with good accurary. 
\begin{figure}[htbp]
    \centering
    \subfigure[the mass density $\rho$ for $\varepsilon=10^{-1}$]{\includegraphics[width=0.425\linewidth]{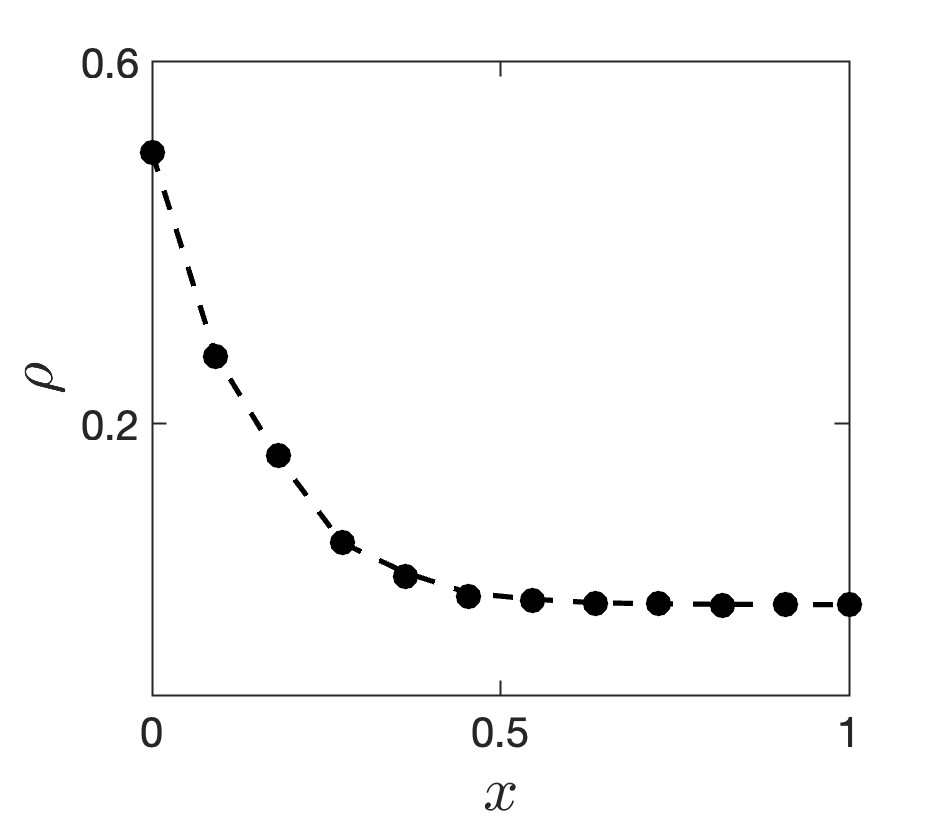}}
    \subfigure[the mass flux $j$ for $\varepsilon=10^{-1}$]{\includegraphics[width=0.425\linewidth]{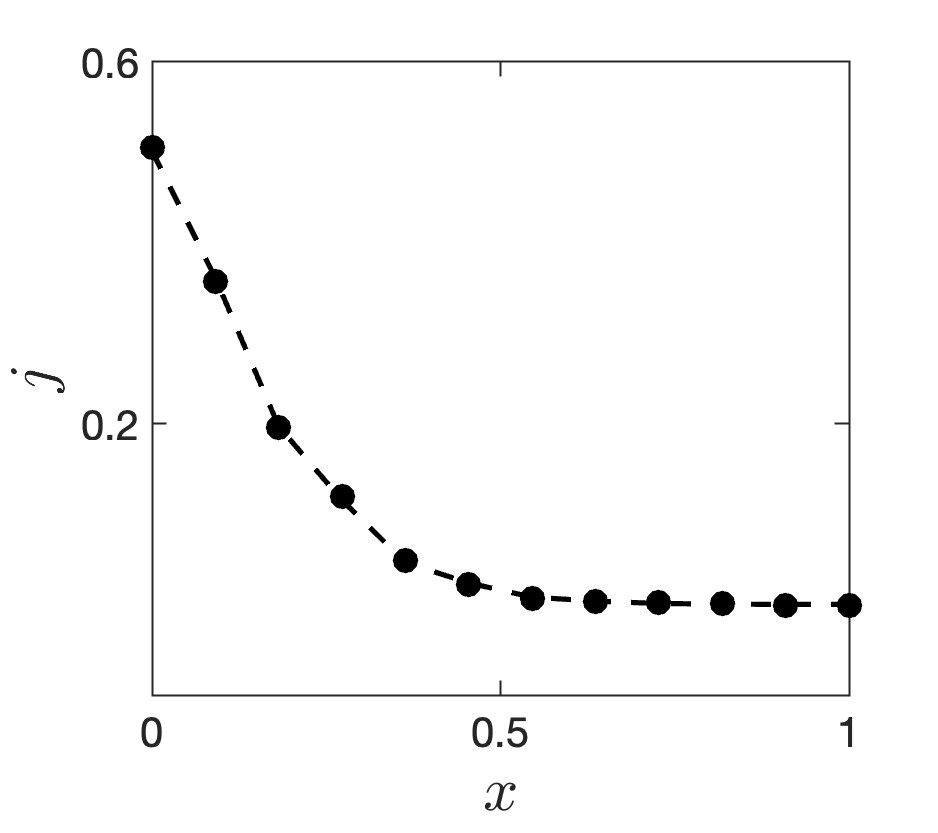}}\\
    \subfigure[the mass density $\rho$ for $\varepsilon=10^{-8}$]{\includegraphics[width=0.425\linewidth]{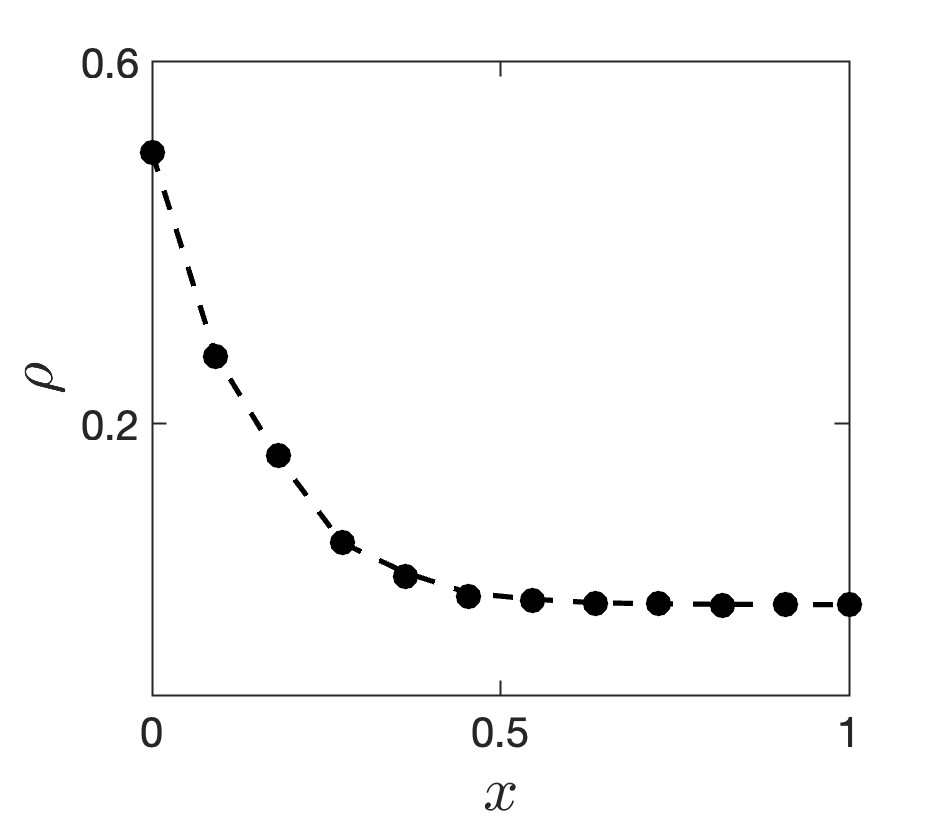}}
    \subfigure[the mass flux $j$ for $\varepsilon=10^{-8}$]{\includegraphics[width=0.425\linewidth]{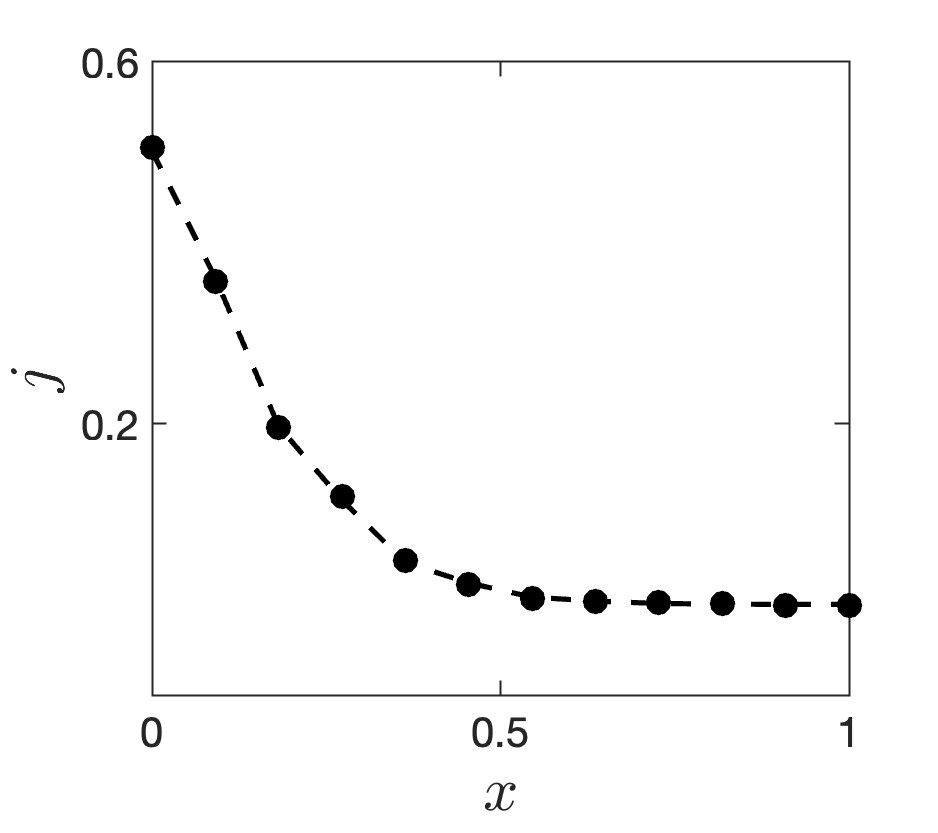}}
    \caption{\label{fig:result:1} Solution at time $t=0.05$ for Problem I using the iterative method, with parameters $\Delta x = 1/10$, $\Delta t = \Delta x^2$ and $N_p = 2^7$.}
\end{figure}
\begin{figure}[htbp]
    \centering
    \subfigure[the mass density $\rho$ for $\varepsilon=10^{-1}$]{\includegraphics[width=0.425\linewidth]{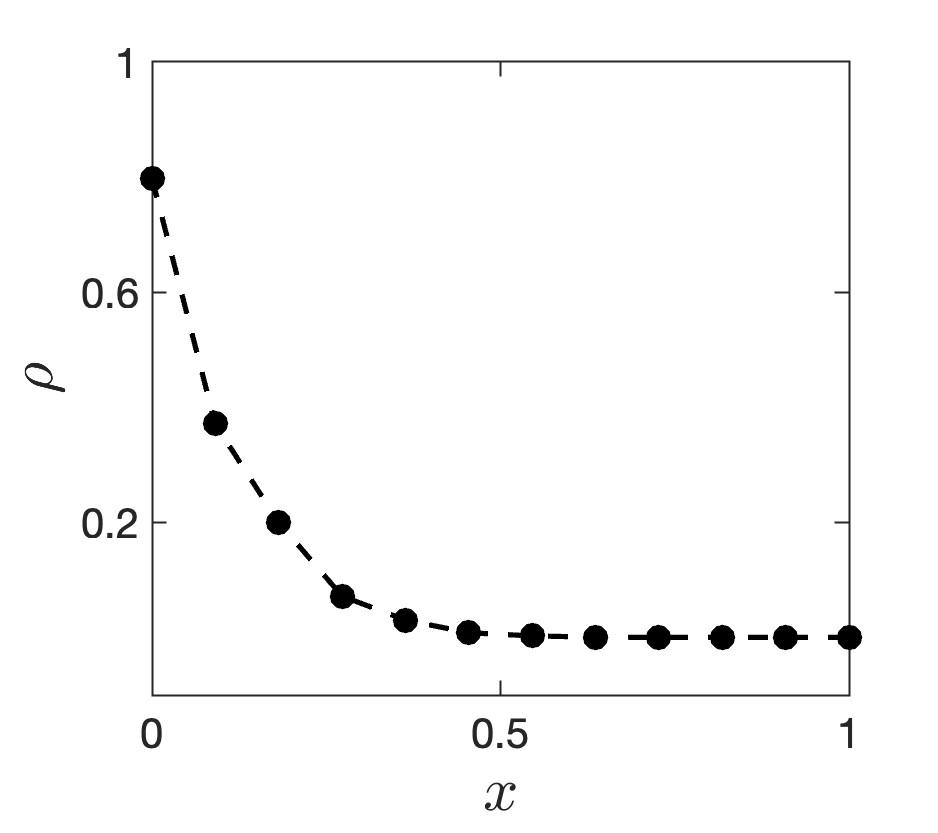}}
    \subfigure[the mass flux $j$ for $\varepsilon=10^{-1}$]{\includegraphics[width=0.425\linewidth]{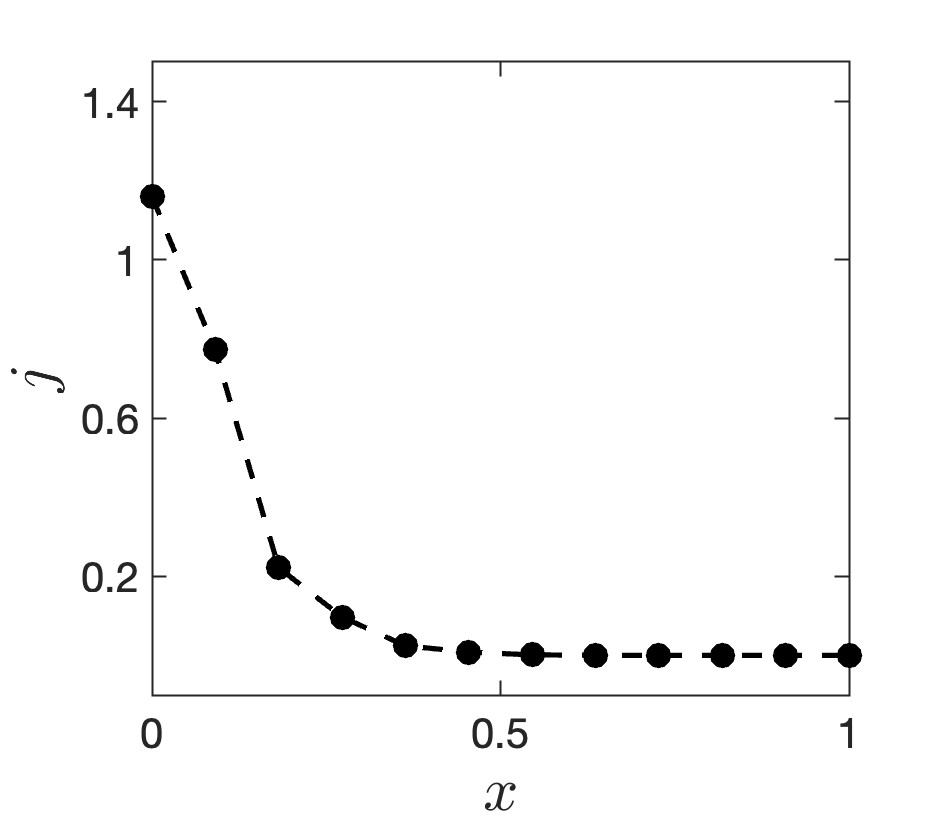}}\\
    \subfigure[the mass density $\rho$ for $\varepsilon=10^{-8}$]{\includegraphics[width=0.425\linewidth]{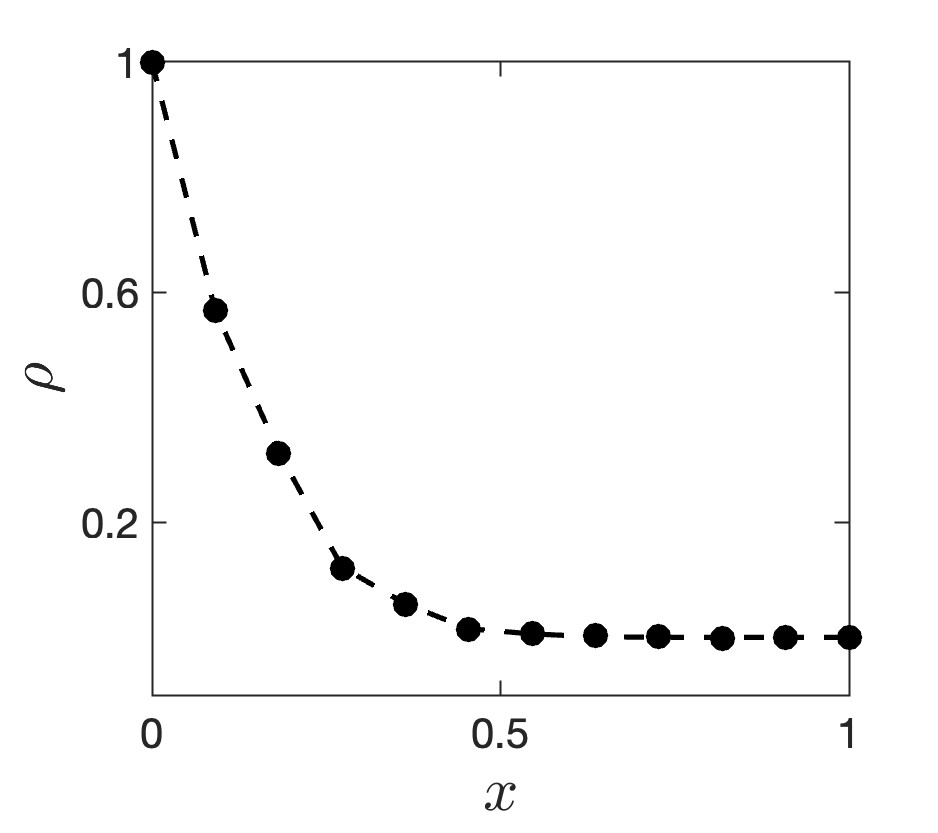}}
    \subfigure[the mass flux $j$ for $\varepsilon=10^{-8}$]{\includegraphics[width=0.425\linewidth]{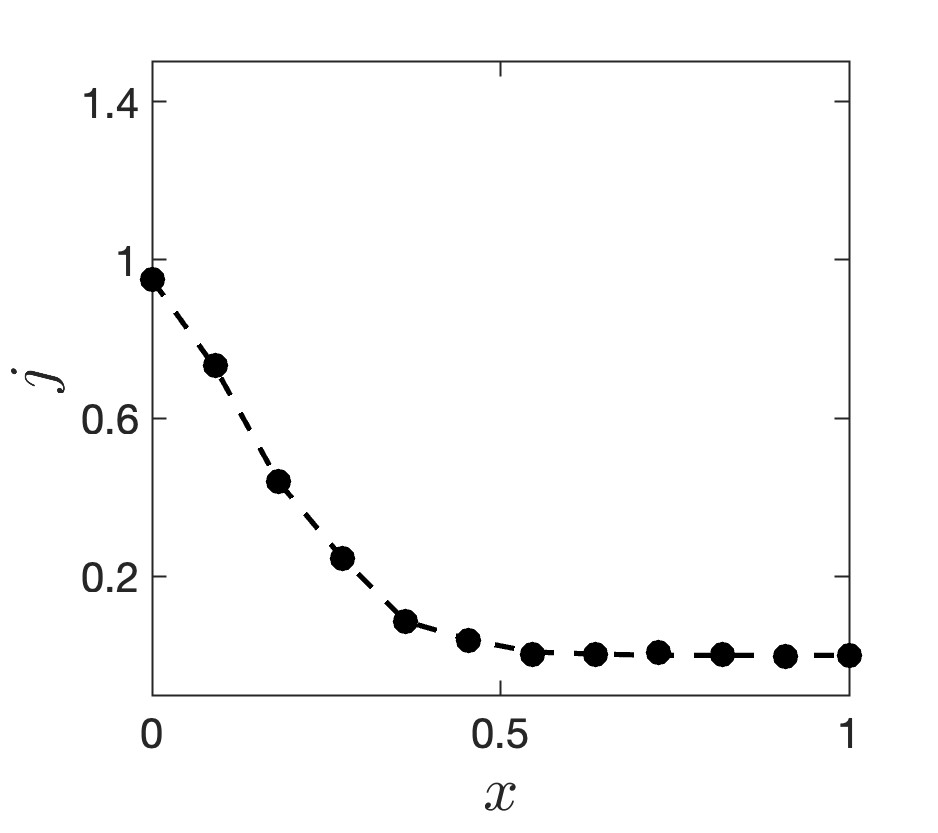}}
    \caption{\label{fig:result:2} Solution at time $t=0.05$ for Problem I using the steady-state solution method, with parameters $\Delta x = 1/10$, $\Delta t = 10/11\Delta x^2$ and $N_p = 2^7$.}
\end{figure}
\begin{figure}[htbp]
    \centering
    \subfigure[the mass density $\rho$ for $\varepsilon=10^{-1}$]{\includegraphics[width=0.425\linewidth]{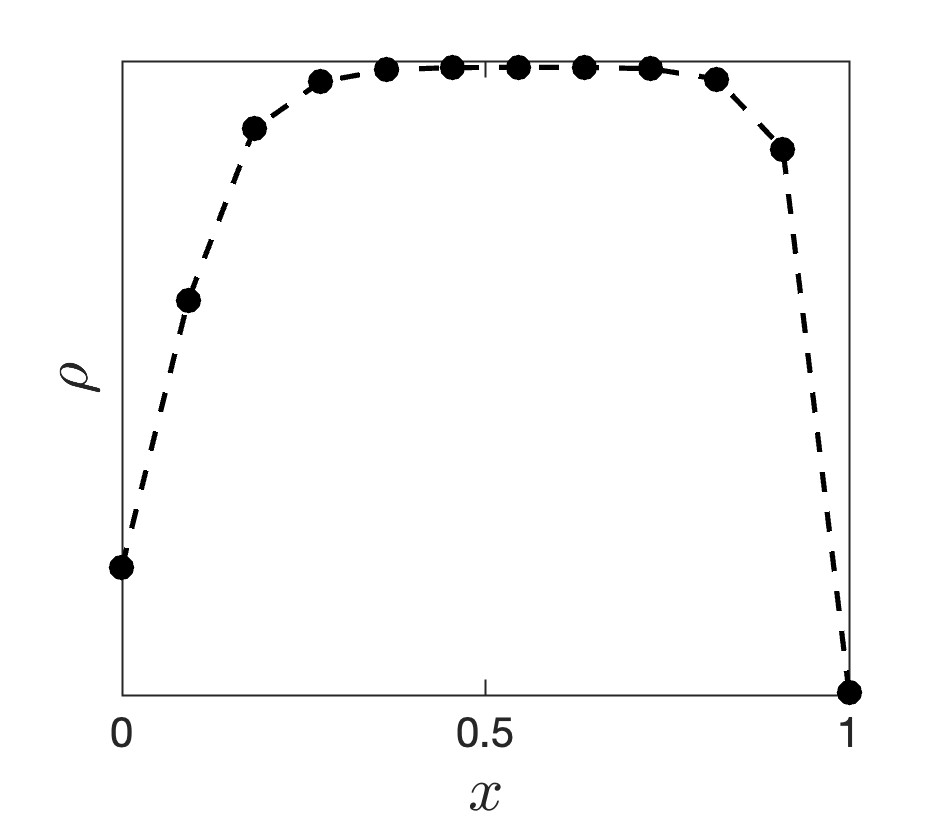}}
    \subfigure[the mass flux $j$ for $\varepsilon=10^{-1}$]{\includegraphics[width=0.425\linewidth]{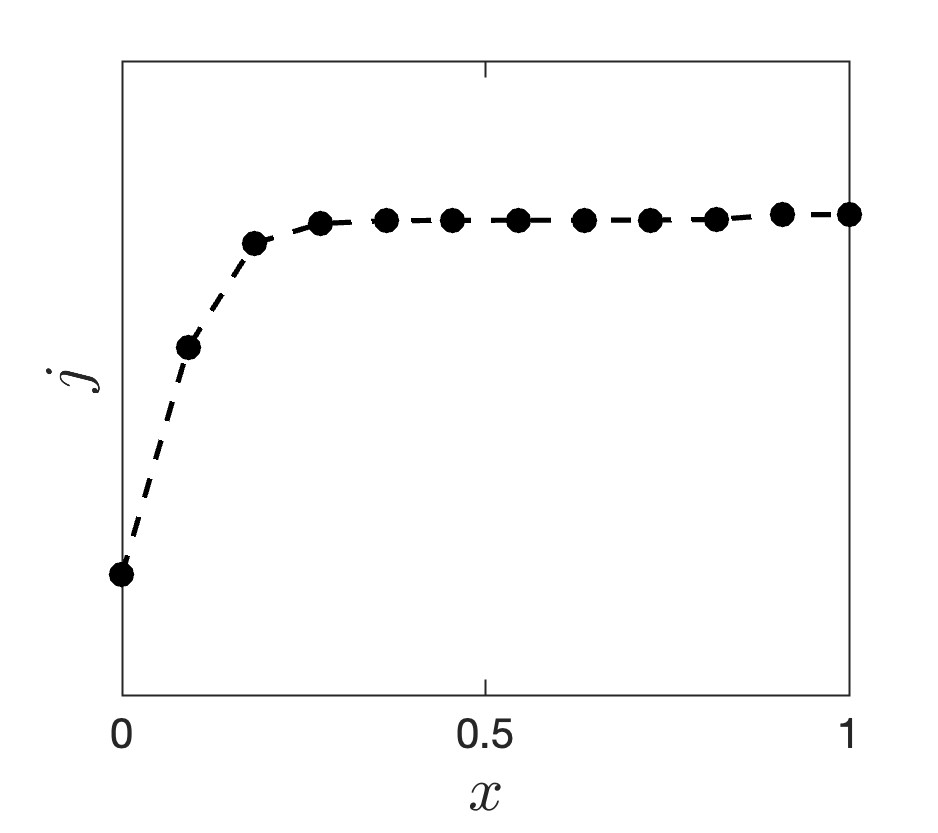}}\\
    \subfigure[the mass density $\rho$ for $\varepsilon=10^{-8}$]{\includegraphics[width=0.425\linewidth]{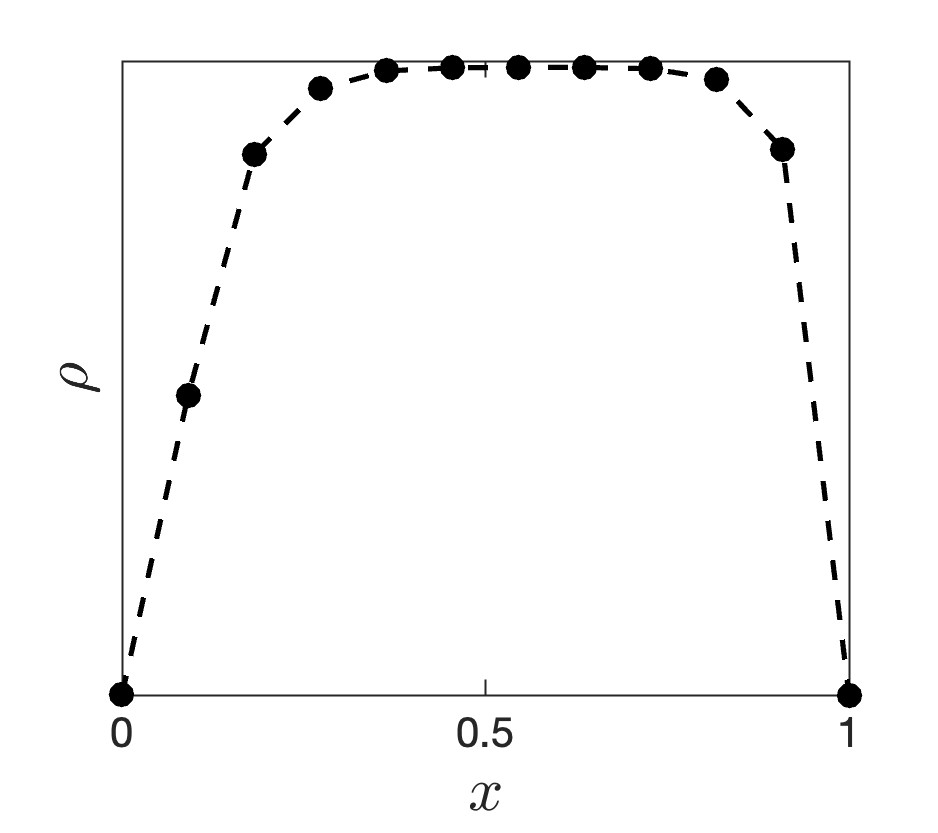}}
    \subfigure[Dthe mass flux $j$ for $\varepsilon=10^{-8}$]{\includegraphics[width=0.425\linewidth]{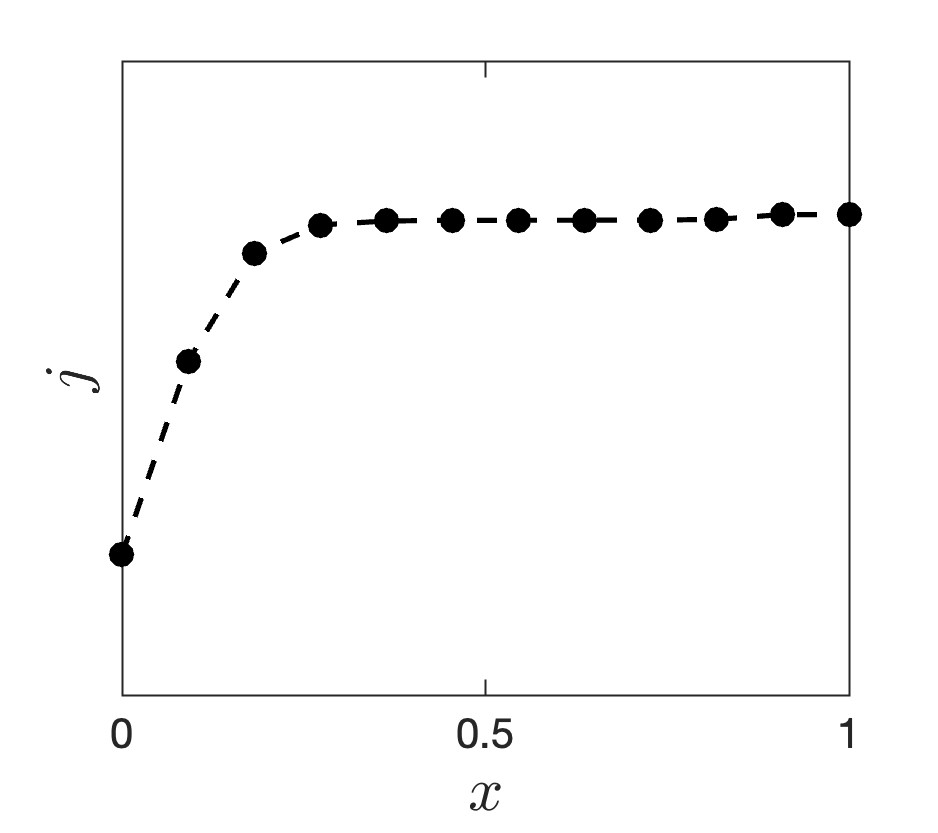}}
    \caption{\label{fig:result:3} Solution at time $t=0.1$ for Problem II using the iterative method, with parameters $\Delta x = 1/10$, $\Delta t = \Delta x^2$ and $N_p = 2^{10}$.}
\end{figure}
\begin{figure}[htbp]
    \centering
    \subfigure[the mass density $\rho$ for $\varepsilon=10^{-1}$]{\includegraphics[width=0.425\linewidth]{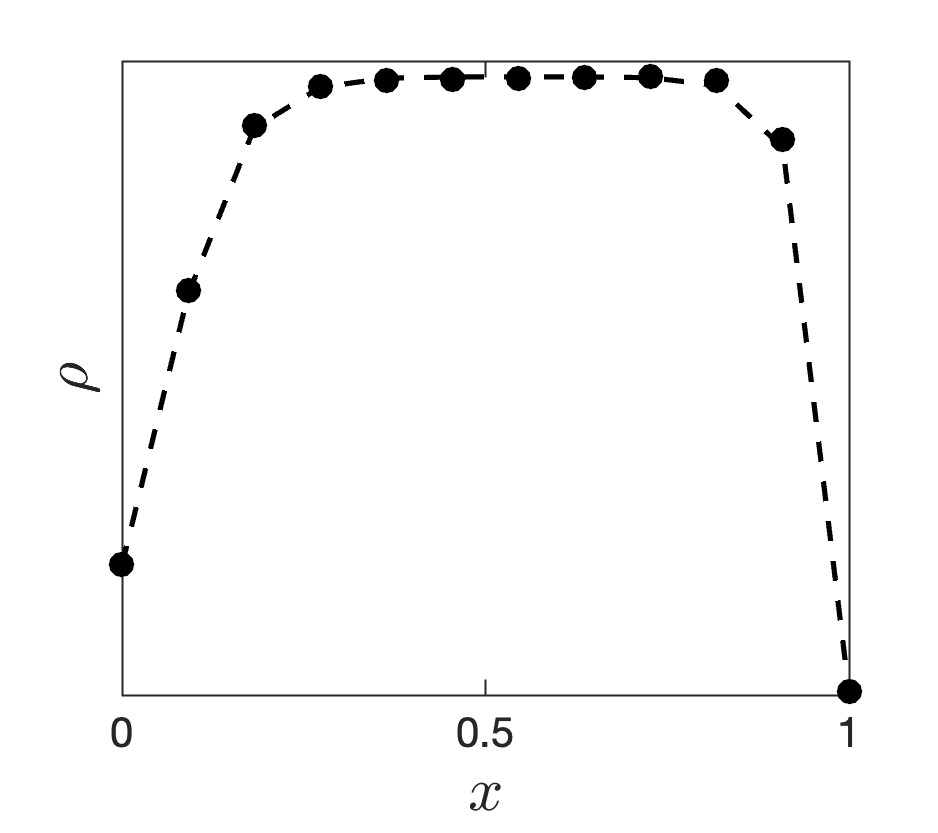}}
    \subfigure[the mass flux $j$ for $\varepsilon=10^{-1}$]{\includegraphics[width=0.425\linewidth]{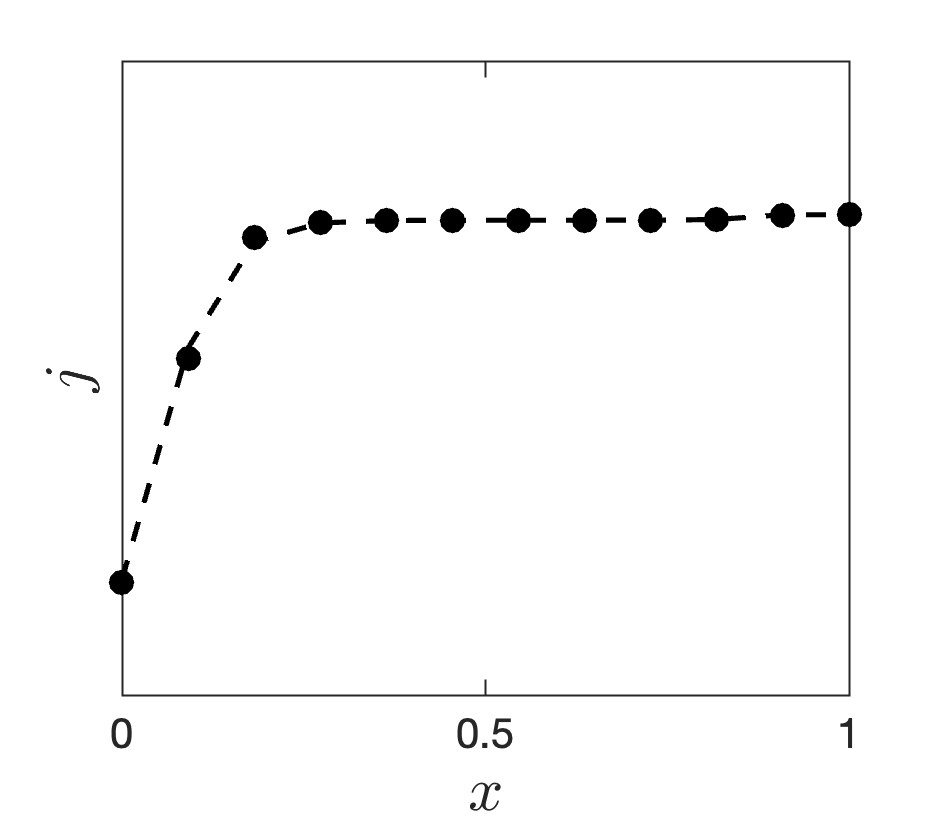}}\\
    \subfigure[the mass density $\rho$ for $\varepsilon=10^{-8}$]{\includegraphics[width=0.425\linewidth]{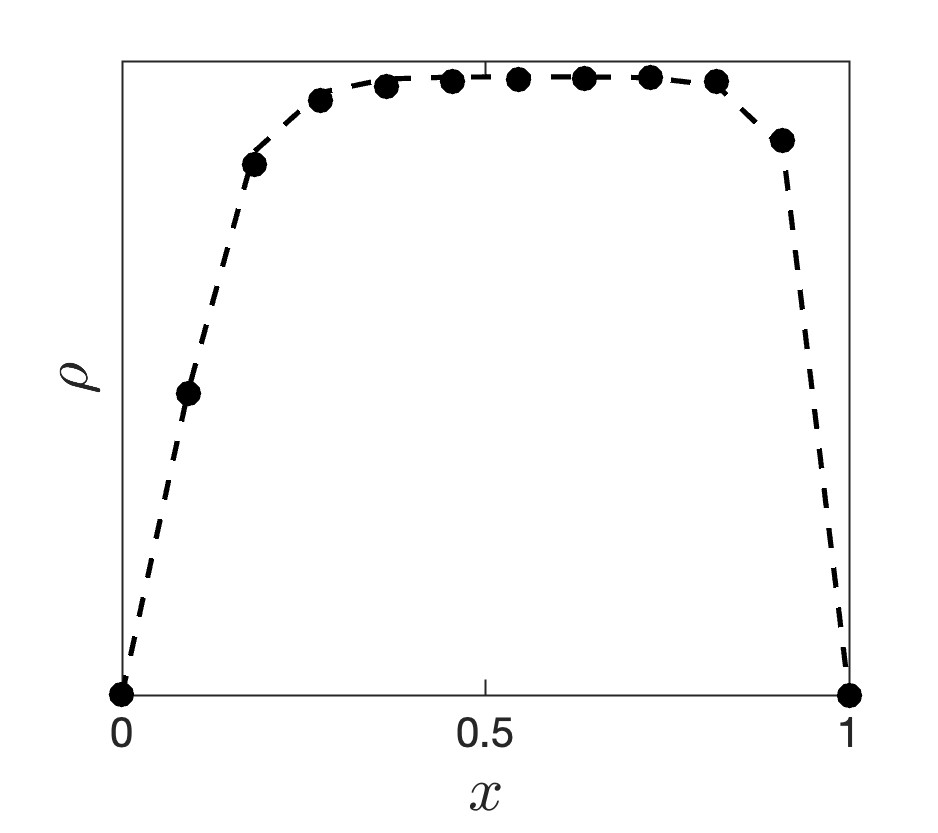}}
    \subfigure[the mass flux $j$ for $\varepsilon=10^{-8}$]{\includegraphics[width=0.425\linewidth]{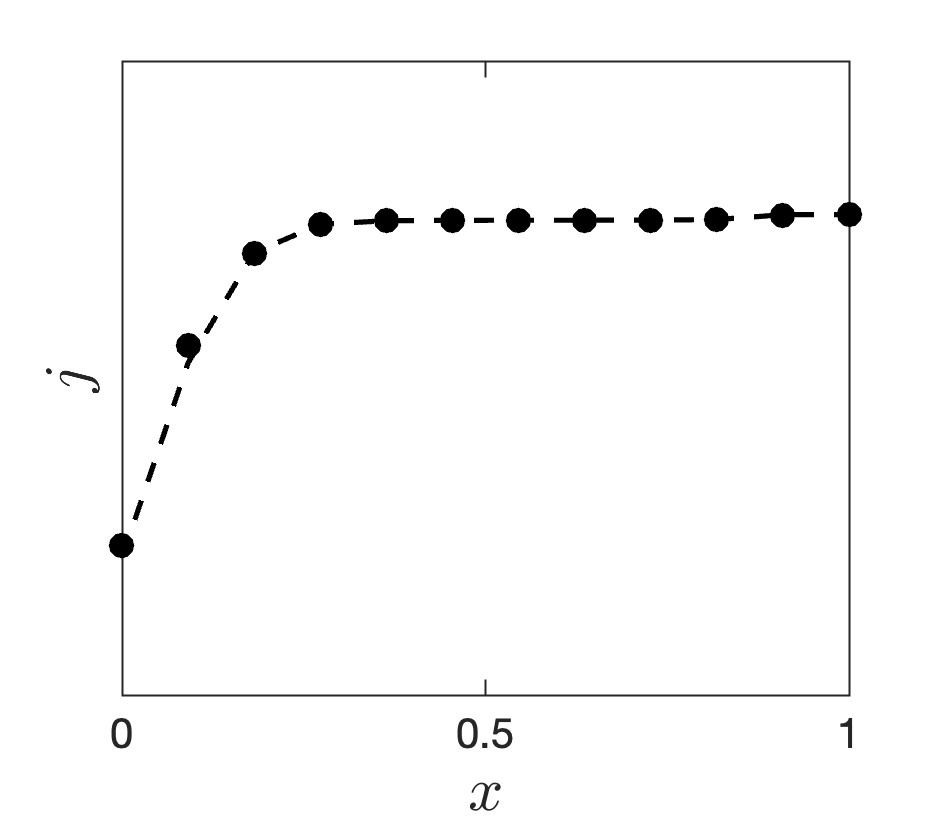}}
    \caption{\label{fig:result:4}Solution at time $t=0.1$ for Problem II using the steady-state solution method, with parameters $\Delta x = 1/10$, $\Delta t = 10/11\Delta x^2$ and $N_p = 2^{9}$.}
\end{figure}
\begin{figure}[htbp]
    \centering
    \subfigure[the mass density $\rho$ for $\varepsilon=10^{-1}$]{\includegraphics[width=0.425\linewidth]{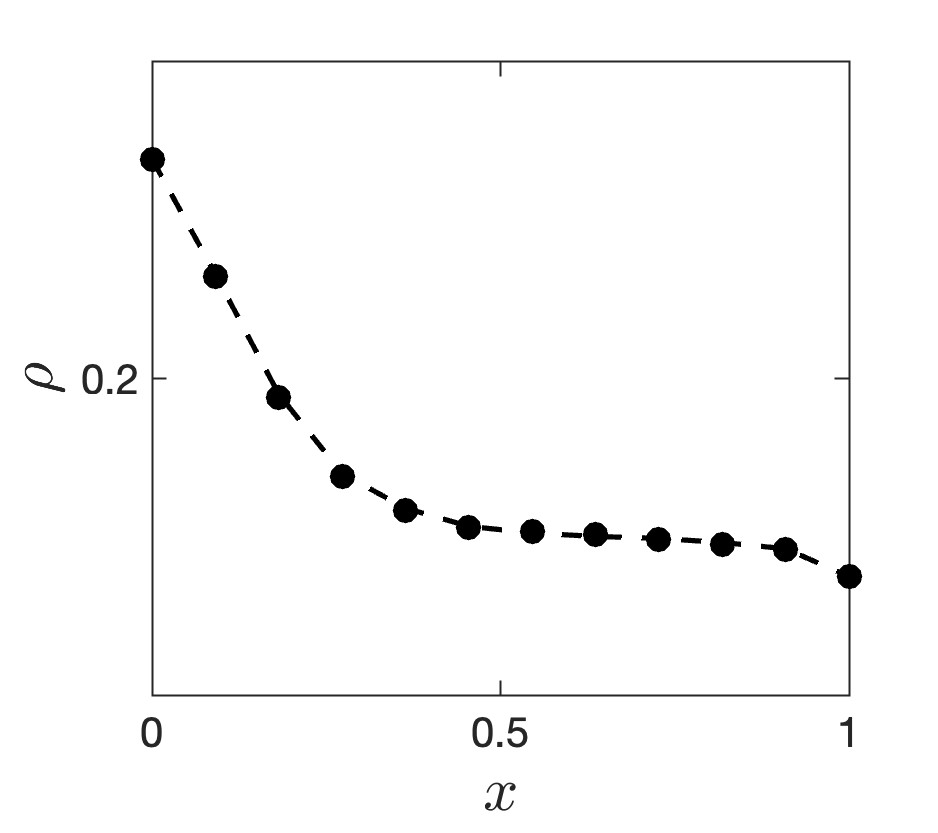}}
    \subfigure[the mass flux $j$ for $\varepsilon=10^{-1}$]{\includegraphics[width=0.425\linewidth]
    {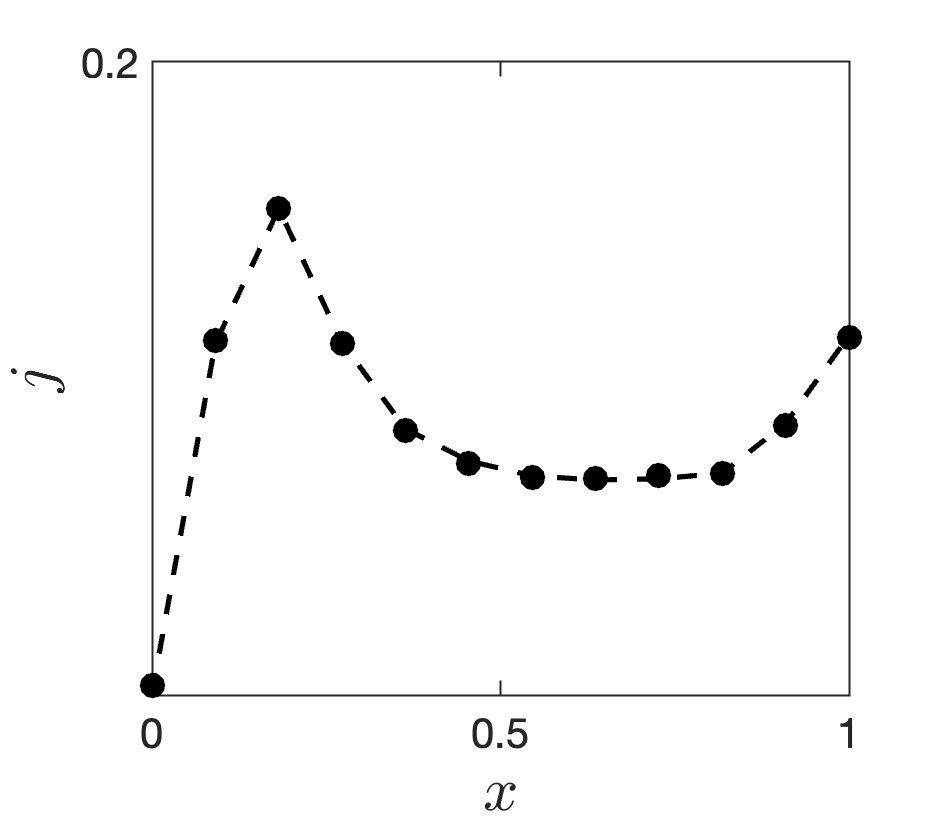}} \\
    \subfigure[the mass density $\rho$ for $\varepsilon=10^{-8}$]{\includegraphics[width=0.425\linewidth]{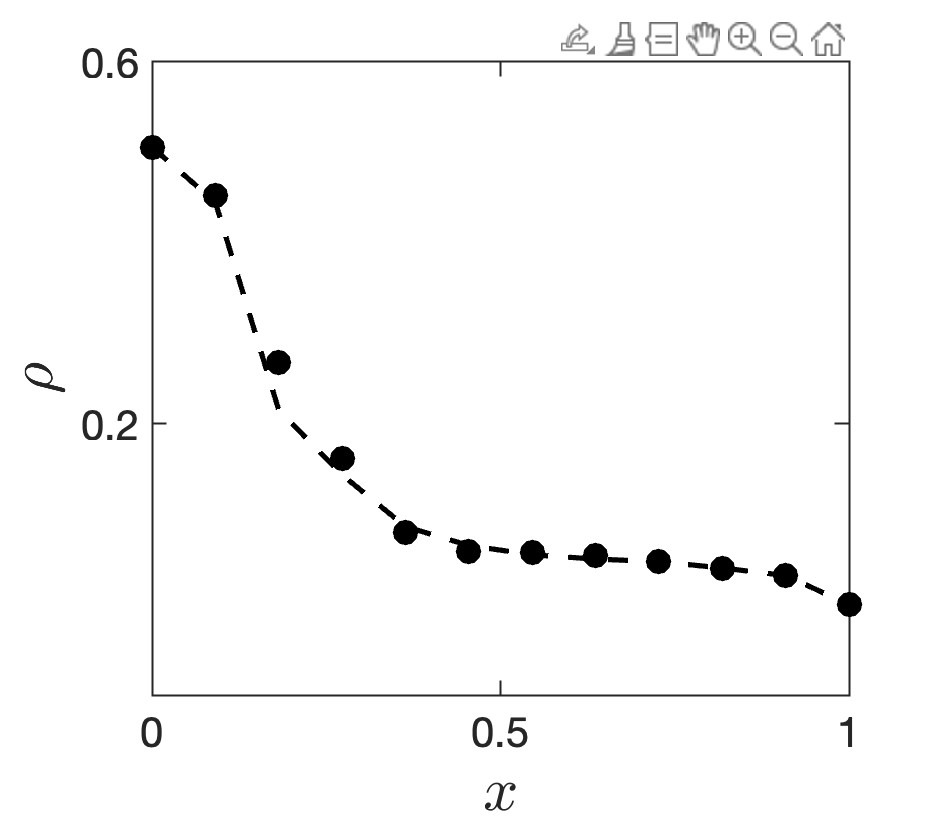}}
    \subfigure[the mass flux $j$ for $\varepsilon=10^{-8}$]{\includegraphics[width=0.425\linewidth]{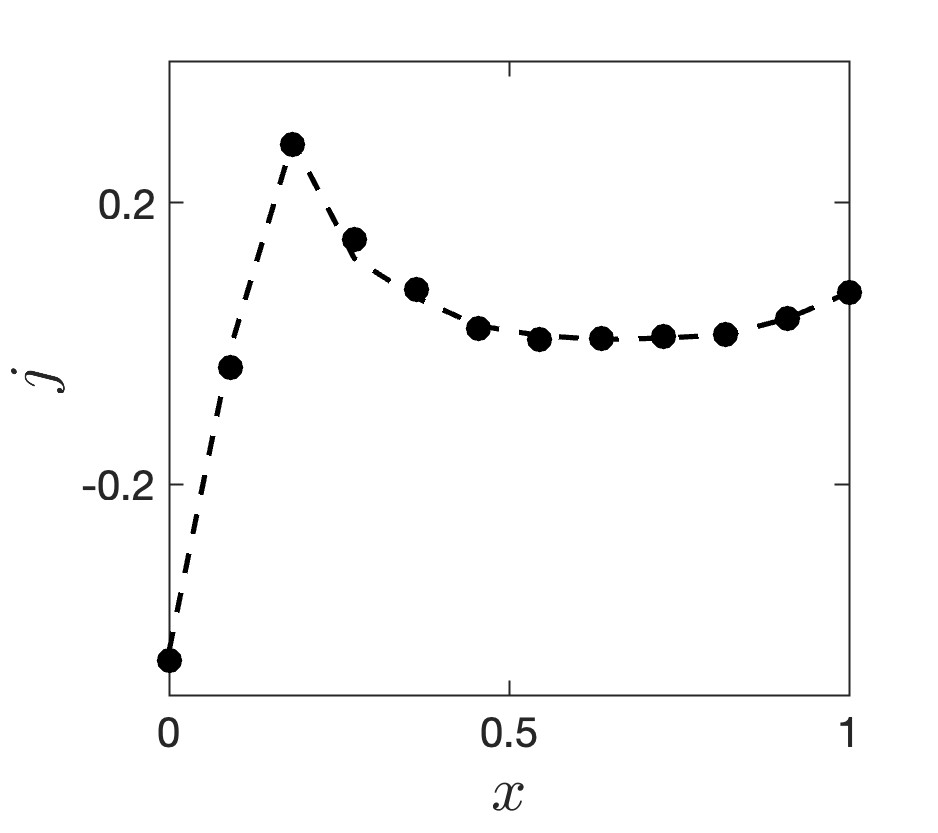}}
    \caption{\label{fig:result:5} Solution at time $t=0.05$ for Problem III using the iterative method, with parameters $\Delta x = 1/10$, $\Delta t = \Delta x^2$ and $N_p = 2^7$.}
\end{figure}
\begin{figure}[htbp]
    \centering
    \subfigure[the mass density $\rho$ for $\varepsilon=10^{-1}$]{\includegraphics[width=0.425\linewidth]{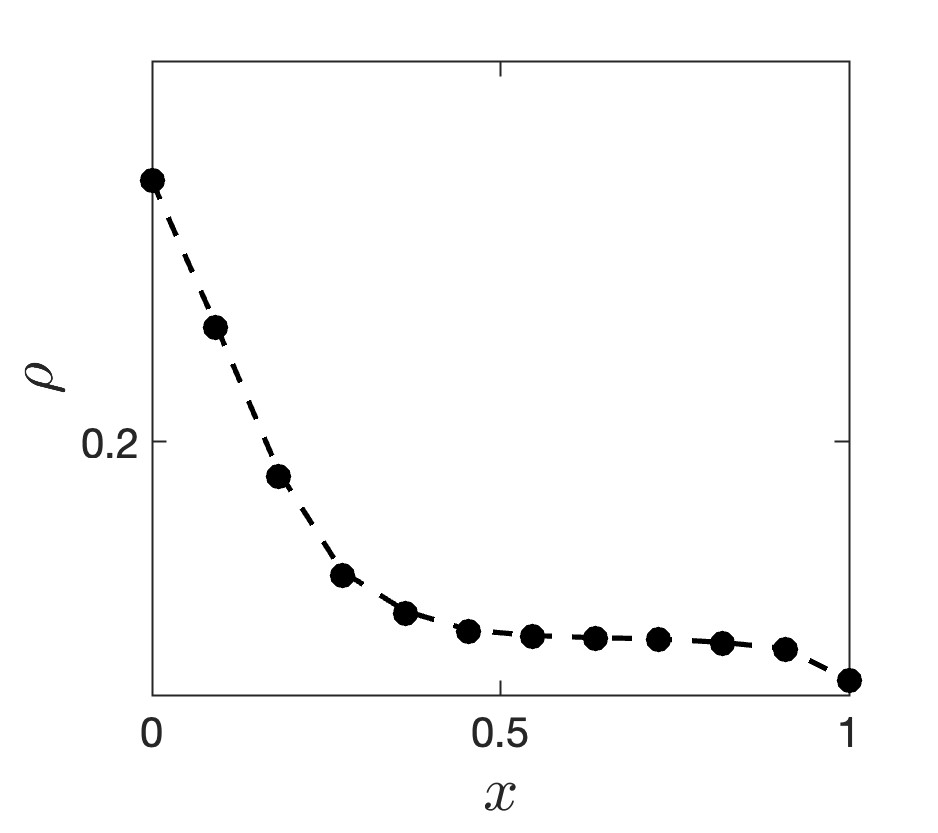}}
    \subfigure[the mass flux $j$ for $\varepsilon=10^{-1}$]{\includegraphics[width=0.425\linewidth]{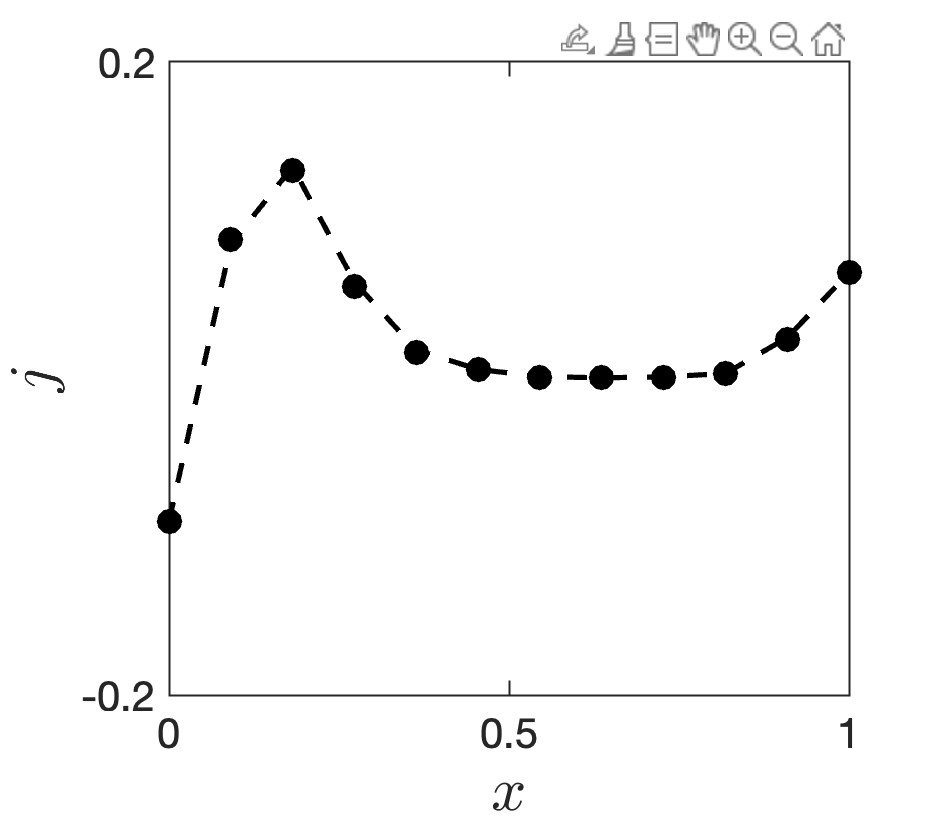}}\\
    \subfigure[the mass density $\rho$ for $\varepsilon=10^{-8}$]{\includegraphics[width=0.425\linewidth]{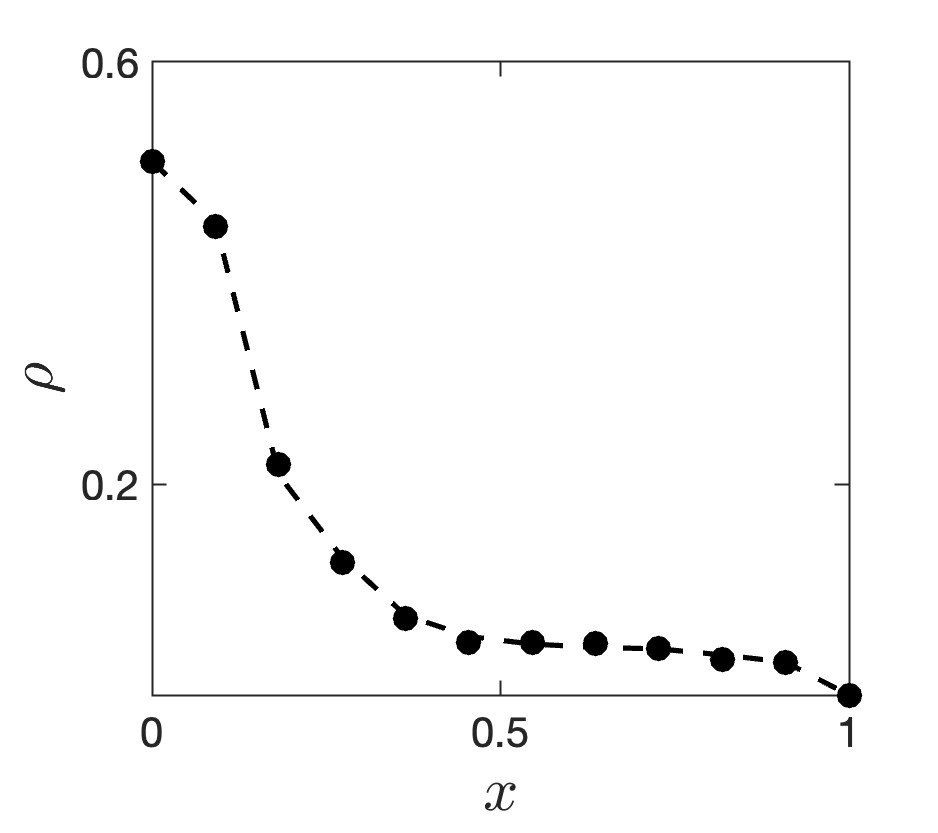}}
    \subfigure[the mass flux $j$ for $\varepsilon=10^{-8}$]{\includegraphics[width=0.425\linewidth]{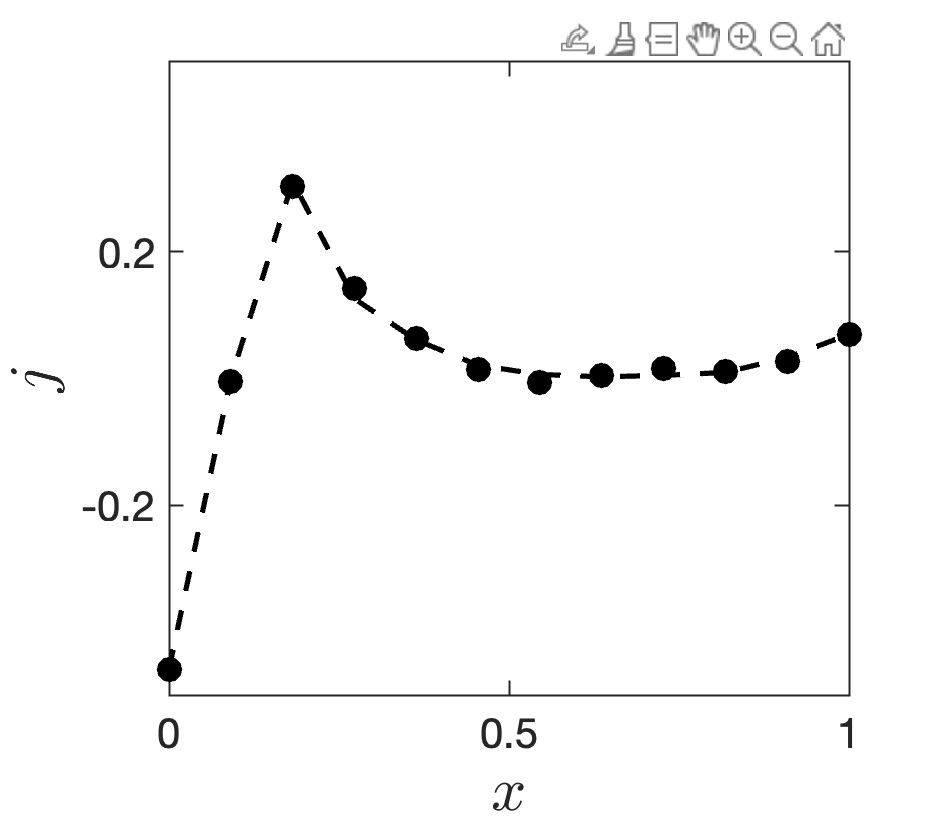}}
    \caption{\label{fig:result:6} Solution at time $t=0.05$ for Problem III using the steady-state solution method, with parameters $\Delta x = 1/10$, $\Delta t = 10/11\Delta x^2$ and $N_p = 2^7$.}
\end{figure}
\section{Conclusions and Discussions}
\label{sec:Conclusions and Discussions}
\par In this paper, we propose two Hamiltonian simulation algorithms for multiscale linear transport equations, both of which utilize Schr\"odingerization and exponential integrator methods and consider the incoming boundary condition. Our algorithmic framework differs from the existing quantum algorithm based on the HHL algorithm proposed by He et al.\cite{Xiaoyang2023TimeCA}, as it is easier to implement on hardware, thus holding greater practical significance. Moreover, the complexities of the two methods we developed: the iterative method presented in Eq. (\ref{theorem:The iterative methods}) and the steady-state solution method presented in Eq. (\ref{theorem:Steady-state solution}), are both $\mathcal{O}(N_vN_x^2\log N_x)$, which is superior than that proposed by He et al\cite{Xiaoyang2023TimeCA} $\mathcal{O}(N_v^2N_x^2\log N_x)$, and provide a polynomial improvement over the $\mathcal{O}(N_v^2N_x^3)$ of traditional algorithms.
\par Our new strategies include the following. First, the preprocessing described in Section \ref{section:preprocessing} effectively handles the term $W$, which is a matrix composed of Gaussian point coefficients, thereby reducing the influence of $N_v$ when estimating query complexity. Then, the iterative method proposed in Section \ref{sec:The iterative methods} takes into account the solution of the state at a specific moment. As a result, we have improved the iterative method based on Schr\"odingerization proposed by Jin et al.\cite{Jin2024QuantumSO} by transforming the approach of calculating evolution time based on fidelity into considering the relationship between the time in the iterative method and the time in the original system. Finally, the steady-state solution method introduced in Section \ref{sec:Steady-state solution} improves upon the Quantum IMEX algorithm proposed by Hu et al.\cite{Hu2025preprint} and incorporates the relaxation-convection scheme into their framework to provide a quantum algorithm for multiscale equations. This allows Hu et al.'s framework to better address multiscale problems or issues involving stiff terms.
\par Although we have only studied multiscale linear transport equations using two frameworks, they can be easily applied to other multiscale problems, such as the multiscale Vlasov-Poisson-Fokker-Planck system. Moreover, although our second method is more complex than the first and uses more quantum bits, its structure can be readily applied to time-dependent PDEs and physical boundary problems, offering greater scalability.

\addcontentsline{toc}{section}{Code Availability}
\section*{Code Availability}
\par The code that support the findings of the main text and the supplement information are will be publicly available upon acceptance.

\addcontentsline{toc}{section}{Declaration of competing interest}
\section*{Declaration of competing interest}
The authors declare that they have no known competing financial interests or personal relationships that could have appeared to influence the work reported in this paper.

\addcontentsline{toc}{section}{Acknowledgement}
\section*{Acknowledgement}
\par SJ was supported by NSFC grant Nos. 12341104 and 12426637, the Shanghai Science and Technology Innovation Action Plan 24LZ1401200,  
the Shanghai Jiao Tong University 2030 Initiative, and the Fundamental Research Funds for the Central Universities. XYH thanks Qitong Hu at Shanghai Jiao Tong University for his advice on eigenvalue analysis.

\addcontentsline{toc}{section}{References}
\bibliographystyle{plain} 


\appendix 
\addcontentsline{toc}{section}{Appendix}
\section{Essential properties and lemmas}
\par First, we present and prove the following eigenvalue perturbation theorem.
\begin{lemma}
\label{lemma:perturbed}
Consider a symmetric matrix $A \in \mathbb{R}^{n \times n}$ and a symmetric perturbed matrix $\varepsilon E$ with $E=\mathcal{O}(1)$, where $\varepsilon$ is a small perturbation parameter. Let $\lambda_i$ be the $i$-th eigenvalue of $A$ and $\xi_i$ be its corresponding normalized eigenvector (i.e., $\|\xi_i\|_2 = 1$), and let $\lambda_{\varepsilon,i}$ be the $i$-th eigenvalue of the perturbed matrix $A+\varepsilon E$, which admits an asymptotic expansion of the form $\lambda_{\varepsilon,i} = \lambda_i+\varepsilon \lambda_i^{(1)} + \mathcal{O}(\varepsilon^2)$ for sufficiently small $\varepsilon$.
Then, the first-order perturbation of the eigenvalue is given by
\begin{equation*}
\lambda_{\varepsilon,i}=\lambda_i+\varepsilon \frac{\xi_i^HE\xi_i}{\xi_i^H\xi_i}+\mathcal{O}(\varepsilon^2),
\end{equation*}
where $\xi_i^H$ denotes the conjugate transpose of $\xi_i$ (which reduces to $\xi_i^T$ for real vectors).
\end{lemma}
\begin{proof}
Using the definition of eigenvalues and eigenvectors, we consider the perturbed eigenproblem:
\begin{equation*}
(A+\varepsilon E)(\xi_i+\varepsilon \xi_i^{(1)}+\cdots)=(\lambda_i+\varepsilon \lambda_i^{(1)}+\cdots)(\xi_i+\varepsilon \xi_i^{(1)}+\cdots),
\end{equation*}
where $\xi_i^{(1)}$ represents the first-order correction to the eigenvector. Expanding this equation and collecting terms of the same order in $\varepsilon$ yields:
\begin{equation}
\begin{aligned}
\label{equ:perturbed:1}
&\mathcal{O}(1): A\xi_i=\lambda_i\xi_i, \quad \text{(original eigenproblem)}\\
&\mathcal{O}(\varepsilon): A\xi_i^{(1)}+E \xi_i =\lambda_i \xi_i^{(1)}+\lambda_i^{(1)}\xi_i, \quad \text{(first-order perturbation)}\\
&\cdots \cdots 
\end{aligned}
\end{equation}
To eliminate $\xi_i^{(1)}$, we multiply $\xi_i^H$ on the left of the $\mathcal{O}(\varepsilon)$ equation and use the symmetry of $A$ ($A^T = A$) and the original eigenproblem: $\xi_i^H A \xi_i^{(1)} + \xi_i^H E \xi_i = \lambda_i \xi_i^H \xi_i^{(1)} + \lambda_i^{(1)} \xi_i^H \xi_i$.
Since $\xi_i^H A = \lambda_i \xi_i^H$, the terms involving $\xi_i^{(1)}$ cancel out, leaving: $\xi_i^H E \xi_i = \lambda_i^{(1)} \xi_i^H \xi_i$. Solving for $\lambda_i^{(1)}$ gives $\lambda_i^{(1)} = \frac{\xi_i^HE\xi_i}{\xi_i^H\xi_i}$, and thus the perturbed eigenvalue becomes $\lambda_{\varepsilon,i}=\lambda_i+\varepsilon \frac{\xi_i^HE\xi_i}{\xi_i^H\xi_i}+\mathcal{O}(\varepsilon^2)$ and this completes the proof.
\end{proof}
\begin{lemma}
\label{lemma:appendix:A:2}
\par Consider a partitioned matrix consisting of $m$ blocks arranged as:
\begin{equation*}
    \begin{aligned}
        B=\begin{pmatrix}
        B_{11} & B_{12} & \cdots & B_{1m} \\
        B_{21} & B_{22} & \cdots & B_{2m} \\
        \vdots & \vdots & \ddots & \vdots \\
        B_{m1} & B_{m2} & \cdots & B_{mm}
        \end{pmatrix},
    \end{aligned}
\end{equation*}
where each $B_{ij}$ is a square submatrix. The singular values of this matrix obey:
\begin{equation}
    \begin{aligned}
        \label{equ:appendix:A:2}
        \| B\|\leq\sum_{k=-(m-1)}^{m-1}\sup_{j-i=k}\| B_{ij}\|,
    \end{aligned}
\end{equation}
with $\| \cdot\|$ representing any operator norm.
\qed
\end{lemma}
\begin{lemma}
\label{lemma:appendix:B:2}
\par The matrices $L_h$ and $D_h$ possess the following properties.
\begin{itemize}
    \item Regarding $L_h$, its eigenvalues take the form $-2+2\sin\left(\frac{\pi(2k-1)}{2(N_x+1)}\right)$ for $k=1,\dots,N_x$. 
    The dominant eigenvalue equals $-2+2\sin\left(\frac{\pi}{2(N_x+1)}\right)$ and remains negative. The smallest eigenvalue is $-2+2\sin\left(\frac{(2N_x-1)\pi}{2(N_x+1)}\right)$
    and exceeds $-4$.
    \item 
    For $D_h$, the square $D_h^2$ can be expressed as
    \begin{equation*}
        \begin{aligned}
            D_h^2 =
            \begin{pmatrix}
            -1 & 0 & 1 &&\\
            0 & -2 & 0 &&\\
            1 & 0 & -2 && \ddots \\
            &&& \ddots && 1 \\
            && \ddots & & -2 & 0\\
            &&& 1 & 0 & -1
            \end{pmatrix}.
        \end{aligned}
    \end{equation*}
    The spectrum of $D_h$ consists of $-2i\sin\left(\frac{\pi(2k-1)}{2(N_x+1)}\right)$, $k=1,\dots,N_x$. Thus, $D_h^2$ has eigenvalues $-4\sin^2\left(\frac{\pi(2k-1)}{2(N_x+1)}\right)$, $k=1,\dots,N_x$, 
    where the largest eigenvalue of $D_h^2$ is non-positive, and the smallest eigenvalue is strictly greater than $-4$.
\end{itemize}
\qed
\end{lemma}
\begin{theorem}
\label{lemma:appendix:B:3}
\par \textbf{(Weyl's inequality).} For any $n\times n$ Hermitian matrices $X$ and $Y$, define $Z=X+Y$. Let their eigenvalues be ordered non-decreasingly as: $\lambda_1(X)\le\cdots\le\lambda_n(X)$, $\lambda_1(Y)\le\cdots\le\lambda_n(Y)$, and $\lambda_1(Z) \le\cdots\le\lambda_n(Z)$. Then, for every $k=1,\cdots,n$, these bounds apply:
\begin{equation*}
    \begin{aligned}
        \lambda_k(X)+\lambda_1(Y)\le\lambda_k(Z)\le\lambda_k(X)+\lambda_n(Y).
    \end{aligned}
\end{equation*}
\qed
\end{theorem}
\section{Error analysis for the 2-norm of the matrix exponential}
\label{section:error:analysis}
\par In this section, we analyze the error between the two matrix exponentials $\|e^{-HT}\|_2$ and $\|e^{-\hat{H}T}\|_2$, where $\hat{H}$ comprises all terms in $H$ excluding those of order $\mathcal{O}(e^{-\tau/\varepsilon^2})$, which are instead entirely contained in the matrix $E$. The proof leverages the Laplace transform and its inverse. Specifically, we begin by evaluating the difference $\left|\|e^{-HT}\|_2 - \|e^{-\hat{H}T}\|_2\right|$ through the resolvent expression $(sI - H)^{-1} - (sI - \hat{H})^{-1}$:
\begin{equation*}
    \begin{aligned}
        \left|\|e^{-HT}\|_2 - \|e^{-\hat{H}T}\|_2\right|\le\|e^{-HT}-e^{-\bar{H}T}\|_2&=\|\mathcal{L}^{-1}\left[(sI-H)^{-1}-(sI-\bar{H})^{-1}\right]\|_2\\
        &\approx \|\mathcal{L}^{-1}\left[(sI-\bar{H})^{-1}E(sI-\bar{H})^{-1}\right]\|_2.
    \end{aligned}
\end{equation*}
Therefore, we use the inverse Laplace transform. First, we recognize that $\mathcal{L}[(sI - \bar{H})^{-1}] = e^{-\bar{H}T}$ and $\mathcal{L}[E(sI - \bar{H})^{-1}] = Ee^{-\bar{H}T}$, which is a relationship where the variables after Laplace transformation are multiplied. Thus, we can utilize the convolution property of the Laplace transform to obtain the value after the inverse Laplace transform of this term:
\begin{equation*}
    \begin{aligned}
        \mathcal{L}^{-1}\left[(sI-\bar{H})^{-1}E(sI-\bar{H})^{-1}\right]=\int_0^te^{\bar{H}\tau}Ee^{\bar{H}(T-\tau)}d\tau.
    \end{aligned}
\end{equation*}
We can use the result in Eq. (\ref{equ:exp_H_esi}) to give an upper bound for $\|e^{-HT}-e^{-\bar{H}T}\|_2$ as follows
\begin{equation*}
    \begin{aligned}
        \left|\|e^{-HT}\|_2 - \|e^{-\hat{H}T}\|_2\right|&\le\int_0^T\|e^{\bar{H}\tau}\|_2\cdot\cdot\|e^{\bar{H}(T-\tau)}\|_2d\tau\\
        &\le T\|E\|_2\left(2+\frac{\|\bar{A}_2\|_2}{\|\bar{A}_1\|_2}\right)^2e^{(-1+\|\bar{A}_1\|_2)T},
    \end{aligned}
\end{equation*}
this indicates that the influence of $E$ on the matrix exponential is of $\mathcal{O}(e^{-\frac{\tau}{\varepsilon^2}})$, and $T$ is of $\mathcal{O}(\frac{1}{\tau})$. Therefore, we only need to ensure that $\varepsilon=o(h/\log h^{-1})$.
\section{Bounds for the 2-norm of matrices}
\subsection{Upper bound for the 2-norm of\texorpdfstring{ $\bar{A}_1$}{}}
\label{section:A1:upper}
\par Through calculation, we can determine the specific structure of $\bar{A}_1$ as follows. In order to better solve its 2-norm, we decompose it into the following matrix additions and multiplications:
\begin{equation}
\begin{aligned}
\label{equ:A1:define}
\bar{A}_1 &= (W_d^{\frac{1}{2}}\otimes I)\left[
(I+\frac{\lambda}{2} (V\otimes I)(I\otimes L_h+\frac{\varepsilon-1}{\varepsilon}E_\nabla\otimes I_2))({W}\otimes I)\right.\\
&\quad\left.+\frac{\lambda}{4h}
(V\otimes D_h)({VW}\otimes I)
(I\otimes D_h+E_\nabla\otimes I_1)
\right](W_d^{-\frac{1}{2}}\otimes I)\\
&:=M_1+E_1+M_2+E_1,
\end{aligned}
\end{equation}
in which $M_1$, $E_1$, $M_2$, $E_1$ are defined as
\begin{equation*}
    \begin{aligned}
    M_1&=(W_d^{\frac{1}{2}}\otimes I)\left(\frac{5}{11}I+\frac{\lambda}{2} (VW\otimes L_h)\right)(W_d^{-\frac{1}{2}}\otimes I)=\left(\frac{5}{11}I+\frac{\lambda}{2} (V\otimes L_h)\right)M_3:=\hat{M}_1M_3,\\
    E_1&=(W_d^{\frac{1}{2}}\otimes I)\left(\frac{\lambda(\varepsilon-1)}{2\varepsilon} (VE_\nabla W\otimes I_2)\right)(W_d^{-\frac{1}{2}}\otimes I)=\left(\frac{\lambda(\varepsilon-1)}{2\varepsilon} (VE_\nabla \otimes I_2)\right)M_3:=\hat{E}_1M_3,\\
    M_2&=(W_d^{\frac{1}{2}}\otimes I)\left(\frac{6}{11}I+\frac{\lambda}{4h}(V^2W\otimes D_h^2)\right)(W_d^{-\frac{1}{2}}\otimes I)=\left(\frac{6}{11}I+\frac{\lambda}{4h}(V^2\otimes D_h^2)\right)M_3:=\hat{M}_2M_3,\\
    E_2&=(W_d^{\frac{1}{2}}\otimes I)\left(\frac{\lambda}{4h}(V^2WE_\nabla\otimes D_hI_1)\right)(W_d^{-\frac{1}{2}}\otimes I)=\frac{\lambda}{4h}(V^2\otimes D_hI_1)M_3(E_\nabla\otimes I),
    \end{aligned}
\end{equation*}
where $M_3=W_d^{\frac{1}{2}} \mathbf{1}\mathbf{1}^T W_d^{\frac{1}{2}} \otimes I $. In the simplification, we frequently use that $V$ and $E_\nabla$ are diagonal matrices. If one does not pre-process $W \otimes I$ to $A_1$, the 2-norm of $M_3$ will be extremely large, making it impossible to perform theoretical analysis and solve for the evolution time, and here we can calculate $\|M_3\|_2$ as
\begin{equation*}
    \begin{aligned}
    \|M_3\|_2&=\|W_d^{\frac{1}{2}}\mathbf{1}\mathbf{1}^TW_d^{\frac{1}{2}}\|_2=1.
    \end{aligned}
\end{equation*}
Therefore, in the following derivation of the upper bound for the 2-norm of $\bar{A}_1$, we use the following relaxation upper bound:
\begin{equation}
\label{equ:A1:main}
\|\bar{A}_1\|_2\le\|\hat{M}_1+\hat{E}_1\|_2\cdot\|M_3\|_2+\|M_2+E_2\|_2.
\end{equation}
\par First, we estimate the upper bound of $\|\hat{M}_1 + \hat{E}_1\|_2$. Clearly, both $\hat{M}_1$ and $\hat{E}_1$ are symmetric matrices, so one can use the spectral radius to compute the $2-$norm, which is the largest absolute value among the eigenvalues. Therefore, we only need to focus on computing the largest and smallest eigenvalues. Moreover, if one can ensure that the smallest eigenvalue of $\hat{M}_1 + \hat{E}_1$ is greater than $0$, one can use $\lambda_{\max}(\hat{M}_1 + \hat{E}_1) = \|\hat{M}_1 + \hat{E}_1\|_2$. Next, we use Weyl's inequality to calculate the minimum eigenvalue of $\hat{M}_1 + \hat{E}_1$ that satisfies
\begin{align*}
\lambda_{\min}(\hat{M}_1 + \hat{E}_1) 
&\ge\lambda_{\min}(\hat{M}_1) + \lambda_{\min}(\hat{E}_1)\ge \frac{5}{11}-2\lambda-\frac{\lambda}{2h}>0,
\end{align*}
where if the condition $\frac{\tau}{h^2}<\frac{10}{11(\frac{4}{h}+1)}<\frac{10}{11}$ is satisfied, one can use Weyl's inequality to obtain the following inequality:

\begin{equation}
\begin{aligned}
\label{equ:A1:part1}
&\|\hat{M}_1+\hat{E}_1\|_2=\lambda_{\max}(\hat{M}_1 + \hat{E}_1)
\le\lambda_{\max}(\hat{M}_1) + \lambda_{\max}(\hat{E}_1)\\
&\le \frac{5}{11}-\frac{\lambda}{2} V_{min}\left(1-\cos\left(\frac{\pi}{N_x+1}\right)\right)\lesssim  \frac{5}{11}-\frac{\pi^2}{2}V_{min}\frac{N_x}{N_t(N_x+1)^2},
\end{aligned}
\end{equation}
where $V_{min}=\min\limits_{i=1}^{N_v}v_i$ and we use the fact that $\hat{E}_1$ is a diagonal matrix with all negative elements, meaning its largest eigenvalue is negative.

\par Next, we estimate the upper bound of $\|M_2 + E_2\|_2$. This calculation is slightly more complex because the right-hand side of $E_2$ cannot be decomposed into $M_3$. However, one can observe that $M_2$ is obtained by multiplying the symmetric matrix $\hat{M}_2$ with another symmetric matrix $\hat{M}_3$, and the value of $E_2$ is of $\mathcal{\varepsilon}$. Therefore, one can use perturbation analysis, which involves first solving $\|\hat{M}_2M_3\|_2$ and then solving $\|M_2 + E_2\|_2$. Before this, we first analyze the properties of $M_3$. Since it is evident that $W_d^{\frac{1}{2}}\mathbf{1}$ is a rank-1 matrix, we can let $\hat{M}_3=W_d^{\frac{1}{2}}\mathbf{1}\otimes I$, thereby obtaining $M_3=\hat{M}_3\hat{M}_3^T$. Based on this, one can derive
\begin{equation*}
    \begin{aligned}
        \|\hat{M}_2M_3\|_2&=\lambda_{\max}\left(\hat{M}_2\hat{M}_3\hat{M}_3^T\hat{M}_3\hat{M}_3^T\hat{M}_2^T\right)^{\frac{1}{2}}=\left\|\hat{M}_3^T\hat{M}_2^2\hat{M}_3\right\|^{\frac{1}{2}},
    \end{aligned}
\end{equation*}
thus transforming the 2-norm problem of a non-symmetric matrix into that of a symmetric matrix, and
\begin{equation}
    \begin{aligned}
        \label{equ:A1:eigenvector}
        \|\hat{M}_2M_3\|_2&=\left\|(\mathbf{1}^TW_d^{\frac{1}{2}}\otimes {I})\left(\frac{6}{11}I+\frac{\lambda}{2}\cdot \frac{1}{2h}({V}^2\otimes D_h^2)\right)^2(W_d^{\frac{1}{2}}\mathbf{1}\otimes I)\right\|_2^{\frac{1}{2}}\\
        &=\lambda_{\max}\left(\left(\frac{2\sqrt{5}}{11}I+\frac{1}{4\sqrt{5}}\frac{\lambda}{h}D_h^2\right)^2+\frac{16}{121}I\right)^{\frac{1}{2}},
    \end{aligned}
\end{equation}
where we employ the identity $\mathbf{1}^TW_d^{\frac{1}{2}}V^kW_d^{\frac{1}{2}}\mathbf{1}=\frac{1}{k+1}$ in this step. Similarly, one can add the condition $\lambda_{\min}(A)>0$ to ensure that $\lambda_{\max}(A^2)=\lambda_{\max}^2(A)$, i.e., when the following condition holds
\begin{equation*}
\lambda_{\min}\left(\frac{2\sqrt{5}}{11}I+\frac{1}{4\sqrt{5}}\frac{\lambda}{h}D_h^2\right)
=
\frac{2\sqrt{5}}{11}+\frac{1}{4\sqrt{5}}\frac{\lambda}{h}\lambda_{\min}(D_h^2)
\ge
\frac{2\sqrt{5}}{11}-\frac{1}{\sqrt{5}}\frac{\lambda}{h},
\end{equation*}
that is, when $\frac{\lambda}{h}<\frac{10}{11}$ holds, one can estimate the upper bound of $\|\hat{M}_2M_3\|_2$ as
\begin{equation}
    \begin{aligned}
    \label{equ:A1:part2}
    \|\hat{M}_2M_3\|_2
    &\le \left(\left(\frac{2\sqrt{5}}{11}-\frac{1}{\sqrt{5}}\frac{\lambda}{h}\cos^2\left(\frac{\pi\lfloor \frac{N_x+1}{2}\rfloor}{N_x+1}\right)\right)^2+\frac{16}{121}\right)^{\frac{1}{2}}\lesssim \frac{6}{11}-\frac{\lambda}{3h}\frac{\pi^2}{(N_x+1)^2}.
    \end{aligned}
\end{equation}
Note that here we need $N_x$ to be an even number, which can be well satisfied by quantum computers.
\par Finally, we use the calculated norm $\|M_2\|_2$ combined with perturbation analysis to provide an upper bound for $\|M_2+E_2\|_2$. First, we have $\|M_2+E_2\|_2=\lambda_{\max}(M_2^TM_2+M_2^TE_2+E_2^TM_2)^{\frac{1}{2}}+o(\varepsilon^2)$. Therefore, one can give an estimate for $\|M_2+E_2\|_2$ based on the maximum eigenvalue $\zeta$ of $M_2^TM_2$ as follows:
\begin{equation*}
    \begin{aligned}
    \|M_2+E_2\|_2=\left(\|M_2\|_2+\frac{\zeta^T(M_2^TE_2+E_2^TM_2)\zeta}{\zeta^T\zeta}\right)^{\frac{1}{2}}+o(\varepsilon^2).
    \end{aligned}
\end{equation*}
Therefore, our tasks below are twofold: first, to calculate the expression for $\zeta$, and second, to compute the value of the fraction. Since $M_2^T M_2$ can be expressed as $\hat{M}_3 \hat{M}_3^T \hat{M}_2^2 \hat{M}_3 \hat{M}_3^T$, we first need to recognize that $\hat{M}_3^T \hat{M}_3 = I$. Therefore, we let the eigenvector corresponding to the largest eigenvalue of $\hat{M}_3^T \hat{M}_2^2 \hat{M}_3$ be $\xi$, and thus $\zeta = \hat{M}_3 \xi$. As for $\xi$, we can infer from the solution process of Eq. (\ref{equ:A1:eigenvector}) that it is the eigenvector corresponding to the $\lfloor \frac{N_x + 1}{2} \rfloor$-th eigenvalue of $D_h^2$, i.e., $| \xi_i | = \left| \sin \left( \frac{i \lfloor \frac{N_x + 1}{2} \rfloor \pi}{N_x + 1} \right) \right|$. Therefore, we can obtain the upper bound of the fraction as follows
\begin{align*}
&\frac{\zeta^T(M_2^TE_2+E_2^TM_2)\zeta}{\zeta^T\zeta}
=\frac{\xi^T\hat{M}_3^T(\hat{M_2}E_2+E_2^T\hat{M}_2)\hat{M}_3\xi}{\xi^T\xi}\\
&=\frac{(\mathbf{1}^TW_d^{\frac{1}{2}}\otimes \xi^T)(\hat{M_2}E_2+E_2^T\hat{M}_2)(W_d^{\frac{1}{2}}\mathbf{1}\otimes\xi)}{\xi^T\xi}:\lesssim\frac{B+B^T}{N_x+1},
\end{align*}
where $B=(\mathbf{1}^TW_d^{\frac{1}{2}}\otimes \xi^T)\hat{M_2}E_2(W_d^{\frac{1}{2}}\mathbf{1}\otimes\xi)$ and $\sum\limits_{i=1}^{N_x}\sin\left(\frac{ik\pi}{N_x+1}\right)^2=\frac{N_x+1}{2},(k=1,\cdots,N_x)$. The exact value of $B$ is difficult to solve, but we can obtain its upper bound as follows:
\begin{align*}
B
&=\frac{\lambda\cdot(\mathbf{1}^TW_d^{\frac{1}{2}}E_\nabla W_d^{\frac{1}{2}}\mathbf{1})}{4h}(\mathbf{1}^TW_d^{\frac{1}{2}}\otimes \xi^T)\left(\frac{6}{11}I+\frac{\lambda}{4h}(V^2\otimes D_h^2)\right)(V^2\otimes D_hI_1)(W_d^{\frac{1}{2}}\mathbf{1}\otimes \xi)\\
&=\frac{\lambda\cdot(\mathbf{1}^TW_d^{\frac{1}{2}}E_\nabla W_d^{\frac{1}{2}}\mathbf{1})}{4h}\left(\frac{6}{11}(\mathbf{1}^TW_d^{\frac{1}{2}}V^2W_d^{\frac{1}{2}}\mathbf{1}\otimes \xi^TD_hI_1\xi)+\frac{\lambda}{4h}(\mathbf{1}^TW_d^{\frac{1}{2}}V^4W_d^{\frac{1}{2}}\mathbf{1}\otimes \xi^TD_h^2D_hI_1\xi)\right)\\
&\lesssim \frac{\lambda\varepsilon}{h^2}\left(\frac{2}{11}+\frac{\lambda}{20h}\mu_2\right)|\xi^TD_h\cdot I_1\xi|,
\end{align*}
in which $\mu_2=-4\cos\left(\frac{\lfloor \frac{N_x+1}{2}\rfloor\pi}{N_x+1}\right)^2\sim \frac{1}{(N_x+1)^2}$ is the $\lfloor\frac{N_x+1}{2}\rfloor$-eigenvalue for $D_h^2$, so our upper bound solution for the fraction is transformed into an upper bound estimation for $|\xi^TD_h\cdot I_1\xi|$:
\begin{align*}
|\xi^TD_h\cdot I_1\xi|\le |\xi_{N_x}||\xi_{N_x-1}|+|\xi_1||\xi_2|\lesssim \frac{1}{N_x+1},
\end{align*}
where $\xi_i$ represents the $i$-th value of the vector $\xi$. By combining the above three inequalities, we obtain $\frac{\zeta^T(M_2^TE_2+E_2^TM_2)\zeta}{\zeta^T\zeta}=o(N_x^{-2})$ when the condition $\varepsilon=o(h/\log h^{-1})$ is satisfied. Therefore, one can use the dominant term $\|M_2\|_2$ to estimate $\|M_2+E_2\|_2$.
\par Combining Eqs. (\ref{equ:A1:part1}), (\ref{equ:A1:part2}) and (\ref{equ:A1:main}), one can conclude a upper bound for $\|\bar{A}_1\|_2$ under the condition $\varepsilon=o(h/\log h^{-1})$ and $\tau/h^2<\frac{10}{11}$:
\begin{align}
\label{equ:A1:result}
\|\bar{A}_1\|_2\le 1-\frac{\pi^2}{2}V_{min}\frac{N_x}{N_t(N_x+1)^2}-\frac{\lambda}{3h}\frac{\pi^2}{(N_x+1)^2},
\end{align}
and $1-\|\bar{A}_1\|_2)\succsim \frac{1}{N_t} =
\mathcal{O}(\frac{1}{N_x^2})$.

\subsection{Lower bound for the 2-norm of\texorpdfstring{ $\bar{A}_1$}{}}
\label{section:A1:lower}
\par For the lower bound of the 2-norm of $\bar{A}_1$, one can use the property that the 2-norm is greater than the maximum norm, thereby transforming this lower bound problem into finding a lower bound for a certain value in $\bar{A}_1$. One may select $\frac{5}{11}(M_3)_{ij}$ from $M_1$ presented in Eq. (\ref{equ:A1:define}), where $i, j$ are arbitrary values, which is $\frac{5}{11}(w_iw_j)^{\frac{1}{2}}$. Using the property $\mathbf{1}^TW_d\mathbf{1}=1$, one can deduce that there exists an $i = j$ such that $w_i \ge \frac{1}{N_v}$. Therefore, we obtain a lower bound for the 2-norm of $\bar{A}_1$:
\begin{align}
\label{equ:A1:result2}
\|\bar{A}_1\|_2\gtrsim \frac{1}{N_v}.
\end{align}

\subsection{Upper bound for the 2-norm of\texorpdfstring{ $\bar{A}_2$}{}}
\label{section:A2:upper}
\par For the upper bound estimation of $\|\bar{A}_2\|_2$, we can provide the following structure of $\bar{A}_2$. Similar to the treatment of $\bar{A}_1$, we decompose it into the following matrix:
\begin{equation*}
\begin{aligned}
\bar{A}_2 &= 
-\frac{1}{N_x}
\left[ 
(I+\frac{\lambda}{4h} V\otimes L_h)
(V\otimes I)(W\otimes I)
(I\otimes D_h+E_\nabla\otimes I_1)\right.\\
&\quad\left.+\frac{\lambda}{2}
(V\otimes I)(I\otimes D_h+\frac{\varepsilon-1}{\varepsilon}E_\nabla\otimes I_1)(W\otimes I)
\right]
(W_d^{-\frac{1}{2}}\otimes I)\\
&:=-\frac{1}{N_x}(L_1+F_1+L_2+F_2),
\end{aligned}
\end{equation*}
in which $L_1$, $F_1$, $L_2$, $F_2$ are defined as
\begin{equation*}
\begin{aligned}
L_1&=(I+\frac{\lambda}{4h} V\otimes L_h)
(VW\otimes D_h)(W_d^{-\frac{1}{2}}\otimes I)=(I+\frac{\lambda}{4h} V\otimes L_h)
(V\otimes D_h)(WW_d^{-\frac{1}{2}}\otimes I):=\hat{L}_1L_3,\\
F_1&=(I+\frac{\lambda}{4h} V\otimes L_h)
(VWE_\nabla\otimes I_1)(W_d^{-\frac{1}{2}}\otimes I)\\
&=(I+\frac{\lambda}{4h} V\otimes L_h)
(V\otimes I)(WW_d^{-\frac{1}{2}}\otimes I)(E_\nabla\otimes I_1):=\hat{F}_{11}L_3\hat{F}_{12},\\
L_2&=\frac{\lambda}{2}
(V\otimes D_hW)(W_d^{-\frac{1}{2}}\otimes I)=\frac{\lambda}{2}
(V\otimes D_h)(WW_d^{-\frac{1}{2}}\otimes I):=\hat{L}_2L_3,\\
F_2&=\frac{\lambda(\varepsilon-1)}{2\varepsilon}
(VE_\nabla W\otimes I_1)(W_d^{-\frac{1}{2}}\otimes I)=\frac{\lambda(\varepsilon-1)}{2\varepsilon}
(VE_\nabla\otimes I_1)(WW_d^{-\frac{1}{2}}\otimes I):=\hat{F}_2L_3,
\end{aligned}
\end{equation*}
in which $L_3=(WW_d^{-\frac{1}{2}}\otimes I)$. Therefore, we transform the upper bound estimation problem of $\bar{A}_2$ into the upper bound estimation of $L_1$, $F_1$, $L_2$, and $F_2$, where the 2-norm of $L_3$ can be directly computed:
\begin{equation*}
\|L_3\|_2=\|WW_d^{-\frac{1}{2}}\otimes I\|_2=\lambda_{\max}(\mathbf{1}\mathbf{1}^TW_d\mathbf{1}\mathbf{1}^T)^{\frac{1}{2}}=N_v^{\frac{1}{2}}.
\end{equation*}
\par As for the upper bounds on the 2-norm of the remaining items $\hat{L}_1$, $\hat{F}_{11}$, $\hat{F}_{12}$, $\hat{L}_2$, and $\hat{F}_2$, we can provide a very rough estimate, the specific process of which is as follows:
\begin{equation*}
\begin{aligned}
\|\hat{L}_1\|_2\le (1+\frac{\lambda}{h})+2=\mathcal{O}(1)&,\quad \|\hat{F}_{11}\|_2\le (1+\frac{\lambda}{h})=\mathcal{O}(1),\\
\hat{F}_{12}\le \frac{\varepsilon}{h}=o(1),\quad \hat{L}_2&\le \lambda,\quad \hat{F}_2\le \frac{\lambda}{2h}=\mathcal{O}(1).
\end{aligned}
\end{equation*}
Based on this, we can obtain the upper bound of $\bar{A}_2$ as follows
\begin{equation}
\begin{aligned}
\label{equ:A2:result}
\|\bar{B}_1\|_2\le\|L_1\|_2+\|F_1\|_2+\|L_2\|_2+\|F_2\|_2\lesssim N_v^{\frac{1}{2}}
\end{aligned}
\end{equation}
\end{document}